\documentclass[c]{IEEEtran}
\usepackage[dvipdfm]{graphicx}
\usepackage[usenames,dvipsnames]{color}
\usepackage{amsmath}%
\usepackage{cite}
\usepackage{amsfonts}%
\usepackage{subfigure}
\usepackage{amssymb,bold-extra}%
\usepackage{graphicx}%
\usepackage{multicol}%
\usepackage{verbatim}%
\usepackage{enumerate}
\usepackage{amsthm}
\usepackage{eucal}
\newtheorem{theorem}{Theorem}

\usepackage[center]{caption}
\usepackage{epstopdf}

\begin{document}

\title{A Diversity-Multiplexing-Delay Tradeoff of ARQ Protocols in The Z-interference Channel}
\author{\large Mohamed S. Nafea$^*$, D. Hamza$^\dag$, Karim G. Seddik$^\ddag$, Mohamed Nafie$^*$, and Hesham El Gamal$^\S$\\ [.1in]
\normalsize  \begin{tabular}{c}
$^*$Wireless Intelligent Networks Center (WINC), Nile University, Cairo, Egypt. \\
$^\dag$Physical Sciences and Engineering Division, KAUST, Thuwal, KSA.\\
$^\ddag$Electronics Engineering Department, American University in Cairo, AUC Avenue, New Cairo, Egypt.\\
$^\S$Department of Electrical and Computer Engineering, Ohio State University, Columbus, USA.\\
\footnotesize Email: mohamed.nafea@nileu.edu.eg, doha.hamzamohamed@kaust.edu.sa, kseddik@ieee.org, mnafie@nileuniversity.edu.eg, helgamal@ece.osu.edu\\ \normalsize
\end{tabular}
}
 \maketitle
\begin{abstract}
In this work, we analyze the fundamental performance tradeoff of the single antenna  Automatic Retransmission reQuest (ARQ) Z-interference channel (ZIC). Specifically, we characterize the achievable three-dimensional tradeoff between diversity (reliability), multiplexing (throughput), and delay (maximum number of retransmissions) of two ARQ protocols: A non-cooperative protocol and a cooperative one. Considering no
cooperation exists, we study the achievable tradeoff of the fixed-power split Han-Kobayashi (HK) approach. Interestingly, we demonstrate that if the second user transmits the common part only of its message in the event of its successful decoding and a decoding failure at the first user, communication is improved
over that achieved by keeping or stopping the transmission of both the common and private messages. We obtain closed-form expressions for the achievable tradeoff under the HK splitting. Under cooperation, two
special cases of the HK are considered for static and dynamic decoders. The difference between the two decoders lies in the ability of the latter to dynamically choose which HK special-case decoding to apply. Cooperation is shown to dramatically increase the achievable first user diversity.

\end{abstract}

\section{Introduction}\label{Int}
The Z-Interference channel (ZIC) is the natural information theoretic model for many practical wireless communication systems. For example, in femto-cells where a mobile station communicating with its long-range base station causes interference to the receiver of a short-range femto-cell, known as the ``loud neighbor problem''  \cite{LNP}, the system can be accurately modeled as a ZIC. This work explores the achievable diversity, multiplexing, and delay tradeoff of the {\it{outage limited}} single antenna ARQ Z-interference channel (ZIC) \cite{LNP} in the large signal-to-noise ratio (SNR) asymptote.

The diversity and multiplexing tradeoff (DMT) framework was initiated by Zheng and Tse \cite{TseDiv} in standard Multi-Input Multi-Output (MIMO) channels. EL Gamal \textit{et al.} \cite{ArqHesham} extended Zheng and Tse's work by introducing the use of ARQ in delay-limited single-link MIMO channels. The authors in \cite{ArqHesham} showed that the ARQ retransmission delay can be leveraged to enhance the reliability of the system at a negligible loss of the effective throughput rate. In addition, the authors in \cite{Kambiz} considered cooperative schemes in ARQ networks; either a single relay is dedicated to simultaneously help two multiple access users or two users cooperate in delivering their messages to a destination equipped with two receiving antennas. In particular, we extend the diversity, multiplexing and delay tradeoff studied in \cite{ArqHesham} to the two user single antenna  ARQ ZIC setting for both non-cooperative and cooperative scenarios.

This work first discusses a non-cooperative ARQ protocol under the use of the two-message fixed-power split Han-Kobayashi (HK) approach at the second user transmitter ~\cite{heshamHK,Bolcskei,tuni,me}. We consider a transmission policy that necessitates that the second transmitter transmits only the common part of its message if it receives a positive acknowledgment (ACK) while a negative acknowledgment (NACK) is received at the first transmitter. By considering two special cases of the HK splitting, a common-message-only (CMO) scheme and a treating-interference-as-noise (TIAN) scheme (i.e. only a private message is sent from the second transmitter) \cite{me}, we show the superiority of our transmission policy over the other policies of continuing or stopping the transmission of both the common and private messages under the stated feedback states. We assume that the splitting parameters are determined according to the outage events at the end of the transmission block of the same information message at both users in order to optimize the achievable diversity gain region (DGR) \cite{me}. The channel state information (CSI) is assumed to be perfectly known at the receivers but is unknown at the transmitters. Therefore, we assume that the chosen splitting parameters remain fixed for fixed rates, interference level, and retransmission delay; the second transmitter can only continue or cease the transmission of its common or private message. We obtain closed-form expressions for the achievable tradeoff under the said policy.

Next, we consider a cooperative ARQ scenario where the second transmitter assists the first one in relaying its message in the event of a NACK reception at the first transmitter. The cooperative protocol divides into static decoding and dynamic decoding schemes. Under static decoding, we solve for the achievable tradeoff by tracing the maximum of that achieved using either the CMO or the TIAN schemes considering the relaying scenario. Under dynamic decoding, the decoder of the first user dynamically changes its decoding algorithm according to the channel conditions revealed to it; either to decode the interference of the other user (i.e. CMO scheme) or to treat it as noise (i.e. TIAN scheme). Finally, we show the superiority of the dynamic decoding scheme over the static one. Unlike the work in \cite{me}, we characterize here the achievable tradeoff at each user for a fixed multiplexing gain of the other user.

To highlight the advantage of the ARQ protocols, we adopt in this paper a coherent delay-limited (or quasi-static) block fading channel model where the channel gains are assumed to be fixed over the transmission of the same information message. By doing this, we focus on the ARQ diversity advantage without exploiting temporal diversity.

The rest of the paper is organized as follows. In Section \ref{sysmod}, we describe the system model and notation. Section \ref{non-coop} analyzes the achievable diversity, multiplexing, and delay tradeoff for the non-cooperative ARQ protocol under the use of the HK approach and its two special cases. In Section \ref{coop}, we characterize the three dimensional tradeoff for two different variations of the cooperative ARQ protocol. Section \ref{Con} concludes the paper.

\section{System Model}\label{sysmod}
We consider a two user single antenna  communication system over a Rayleigh fading Z-interference channel (ZIC). User's two transmitter (TX2), causes interference to user's one receiver (RX1) but not vice versa as depicted in Fig. \ref{zmodel}. Both users are backlogged, i.e., they always have information messages to send. Each user in our model employs an ARQ error control protocol with a maximum of $L$ transmission rounds. To allow for retransmissions, the information message from each transmitter is encoded into a sequence of $L$ vectors, $\left\{x_{i,l}\in \mathbb{C}^T :\;\;i=1,2\;\;\mbox{and } l=1,\cdots,L\right\}$, where the transmission of each vector takes $T$ channel uses. Each decoder is allowed to process its corresponding received signal over all the $l$ received blocks to decode the transmitted message. Each receiver sends an ACK back to its corresponding transmitter when decoding is successful. A NACK is sent if decoding fails. The ACK/NACK one-bit message is the only feedback allowed in this model and the ARQ feedback channel is assumed to be error-free and of negligible delay.

Our system prescribes to the following ARQ protocol. When both transmitters receive an ACK, they each proceed to send the first block of their next messages. If TX1 receives an ACK while TX2 receives a NACK, TX1 will cease its transmission until TX2 receives an ACK. When TX1 receives a NACK for its message, it begins the transmission of the next block of its current message; while the behavior of TX2 varies according to its feedback outcome and the used ARQ protocol as detailed in the next sections. The reason for differentiating between the case when TX1 receives an ACK while TX2 receives a NACK and the reverse case is that the first user message is not decoded at the second receiver but not vice versa. When the maximum number of protocol rounds $L$ is reached, both transmitters start transmitting the first block of their next messages regardless of the feedback outcome. Error at each user occurs due to any of the following two events. Either $L$ transmission rounds are reached and decoding fails or the decoder makes a decoding error at round $l\leq L$ and fails to detect it (undetected error event).

Based on the above description, the received signal vectors can be described as follows.
\begin{equation}
\begin{split}
\bold{y}_{1,l}&=\mu_{11}h_{11}\bold{x}_{1,l}+\mu_{21}h_{21}\bold{x}_{2,l}+\bold{n}_{1,l}\\
\bold{y}_{2,l}&=\mu_{22}h_{22}\bold{x}_{2,l}+\bold{n}_{2,l},
\end{split}
\end{equation}
\noindent where $\left\{\bold{y}_{i,l},\bold{n}_{i,l}\right\}$ denote the received vector and the noise vector at RX$i$, respectively. The noise vectors are modeled as complex Gaussian random vectors with i.i.d. entries, i.e., $\bold{n}_{1,l}, \bold{n}_{2,l}\sim{\mathcal{CN}}(\bold{0},\bold{I}_T)$ for $l=1,\cdots,L$. They are also assumed to be temporally white. We use $\left\{h_{i,j}:\;i=1,2 \mbox{ and } j=1,2\right\}$ for the channel gain between transmitter $i$ and receiver $j$. The channel gains are i.i.d complex Gaussian random variables with zero mean and unit variance. They are assumed to remain constant over the $L$ transmission rounds and change to new independent values with each new information message. We use this assumption to quantify the diversity gain of the ARQ protocol when no temporal diversity is exploited. We use a per-block power constraint such that $E\left[\frac{1}{T}||x_{i,l}||^2\right]\leq\rho$, i.e., the constraint on the {\it{average}} transmitted power in each transmission round of the ARQ protocol is the same. The parameter $\rho$ takes on the meaning of {\it{average}} SNR per receiver antenna. We also parameterize the attenuation of transmit signal $i$ at receiver $j$ using the real-valued coefficients $\mu_{ij}>0$. To simplify our results, we set $\mu_{11}^2=\mu_{22}^2=1$ and $\mu_{21}^2=\rho^{\beta-1}$. The parameter $\beta$ represents the interference level, $\beta\geq0$.

The decoder is allowed to process the received signal over the $l$ transmission rounds, hence it is convenient to work instead with the following accumulated received vectors
\begin{equation}
\begin{split}
\bold{\tilde{y}}_{1,l}&=h_{11}\bold{\tilde{x}}_{1,l}+\sqrt{\rho^{\beta-1}}h_{21}\bold{\tilde{x}}_{2,l}+\bold{\tilde{n}}_{1,l}\\
\bold{\tilde{y}}_{2,l}&=h_{22}\bold{\tilde{x}}_{2,l}+\bold{\tilde{n}}_{2,l},
\end{split}
\end{equation}
\noindent where $\bold{\tilde{y}}_{i,l}=[\bold{y}_{i,1},\bold{y}_{i,2},...,\bold{y}_{i,l}]$, and all other vectors above are similarly defined.

\begin{figure}
	\centering
	\includegraphics[width=70mm, height = 40mm]{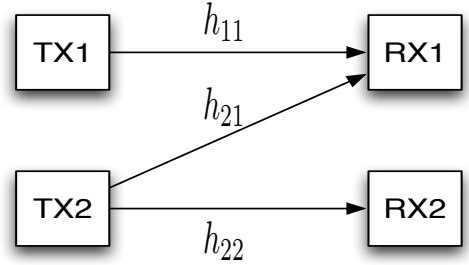}
	\caption{The ZIC model.}
	\label{zmodel}
\vspace{-.2in}
\end{figure}

We consider a family of ARQ protocols that is based on a family of code pairs $\left\{C_1(\rho), C_2(\rho)\right\}$ with first block rates $R_1(\rho)$ and $R_2(\rho)$, respectively, and an overall block length $TL$. The individual error probabilities at RX1 and RX2 are $P_{e_1}(L,\rho)$ and $P_{e_2}(L,\rho)$, respectively. For this family, the first block multiplexing gains $r_1$ and $r_2$ are defined as

\begin{equation}
r_1\triangleq\lim_{{\rho}\rightarrow\infty}\frac{R_1(\rho)}{\log\rho}\qquad\text{and}\qquad r_2\triangleq\lim_{{\rho}\rightarrow\infty}\frac{R_2(\rho)}{\log\rho}.
\end{equation}

Also, the effective ARQ diversity gains at RX1 and RX2 for $L$ transmission rounds are defined as
\begin{equation}
\begin{split}
d_1(L)&\triangleq-\lim_{{\rho}\rightarrow\infty}\frac{\log\left\{P_{e_1}(L,\rho)\right\}}{\log{\rho}}\\
d_2(L)&\triangleq-\lim_{{\rho}\rightarrow\infty}\frac{\log\left\{P_{e_2}(L,\rho)\right\}}{\log{\rho}}.
\end{split}
\end{equation}

Based on the above discussion, we now characterize the long-term average throughput of the ARQ protocol $\eta_1$ for TX1 and $\eta_2$ for TX2. Define a r.v. $\zeta$ as the time between two successive events of sending new information messages by both users. Let ${\cal{A}}_{l}$ denote the event that an ACK is fed back at round $l$ from RX1 and let ${\cal{B}}_{l}$ denote the event that an ACK is fed back at round $l$ from RX2, for $l=1,...,L-1$. Also, let $\overline{\cal{A}}_l$ and $\overline{\cal{B}}_l$ denote the complement events. Thus, we have for $l=1,\cdots,L-1$
\begin{equation}
{\rm{Pr}}\left\{\zeta>l\right\}={\rm{Pr}}\left\{\left(\overline{\cal{A}}_1,\cdots,\overline{\cal{A}}_l\right)\cup\left(\overline{\cal{B}}_1, \cdots,\overline{\cal{B}}_{l}\right)\right\}.
\label{eq:thruput1}
\end{equation}
The expected time between two successive events of sending new information messages (in slots) can be written as
\begin{equation}
\mathbb{E}(\zeta)=\sum_{l=0}^{L-1}{\rm{Pr}}\left\{\zeta>l\right\},
\end{equation}
where, by definition, ${\rm{Pr}}\left\{\zeta>0\right\}=1$. Thus, we have
\begin{equation}
\mathbb{E}(\zeta)= 1+\sum_{l=1}^{L-1}{\rm{Pr}}\left\{\zeta>l\right\}.
\label{}
\end{equation}
Using equation (\ref{eq:thruput1}) and considering the union bound, we have
\begin{equation}
\begin{split}
\mathbb{E}(\zeta)&\leq 1+\sum_{l=1}^{L-1}{\rm{Pr}}\left\{\overline{\cal{A}}_1,\cdots,\overline{\cal{A}}_l\right\}+{\rm{Pr}}\left\{\overline{\cal{B}}_1,\cdots,\overline{\cal{B}}_l\right\}\\
&\leq 1+\sum_{l=1}^{L-1}{\rm{Pr}}\left\{\overline{\cal{A}}_l\right\}+{\rm{Pr}}\left\{\overline{\cal{B}}_l\right\}.
\end{split}
\end{equation}

The average throughput of the ARQ protocol $\eta_1$ for TX1 and $\eta_2$ for TX2 can be characterized as follows.
\begin{equation}
\begin{split}
\eta_1&=\frac{R_1(\rho)}{\mathbb{E}(\zeta)}\\
&\geq\frac{R_1(\rho)}{1+\sum_{l=1}^{L-1}{\rm{Pr}}\left\{\overline{\cal{A}}_l\right\}+{\rm{Pr}}\left\{\overline{\cal{B}}_l\right\}}\\
\eta_2&=\frac{R_2(\rho)}{\mathbb{E}(\zeta)}\\
&\geq\frac{R_2(\rho)}{1+\sum_{l=1}^{L-1}{\rm{Pr}}\left\{\overline{\cal{A}}_l\right\}+{\rm{Pr}}\left\{\overline{\cal{B}}_l\right\}},
\label{eq:throughput}
\end{split}
\end{equation}
Also, by definition, the following relations hold.
\begin{equation}
\eta_1\leq R_1(\rho)\;\;\;\text{and}\;\;\;\eta_2\leq R_2(\rho).
\label{eq:throughput11}
\end{equation}
Then the effective multiplexing rates are defined as
\begin{equation}
\begin{split}
r_{e_1}\triangleq\lim_{\rho\rightarrow\infty}\frac{\eta_{1}(\rho)}{\log(\rho)}\qquad\text{and}\qquad r_{e_2}\triangleq\lim_{\rho\rightarrow\infty}\frac{\eta_{2}(\rho)}{\log(\rho)}.
\end{split}
\end{equation}

Throughout the paper, we use a scheme subscript to distinguish between the diversity gains under the different scenarios considered. So, for example, we use $d_{1,{\rm HK}}(L)$ to denote the diversity gain at RX1 under the HK scheme.

\section{The Non-Cooperative ARQ Protocol}\label{non-coop}
We investigate here the diversity, multiplexing, and delay tradeoff here under the use of a non-cooperative ARQ protocol. We consider the use of the two-message HK scheme at TX2. Specifically, TX2 maintains a private message with rate $S_2=s_2\log\rho$ and a common message with rate $T_2=t_2\log\rho$. Hence, $r_2=s_2+t_2$, $s_2,t_2\geq 0$, and $0\leq r_i\leq 1$. At RX1, we consider a joint typical-set decoder applied to the message of TX1 and the common message of TX2. At RX2, jointly-typical set detection is carried out for both the private and common messages of TX2.  The reason of using joint typical-set decoding here is the need for error detection capabilities to prove achievability of our results. For TX2, we parameterize the ratio of the average private power to the total average power as
\begin{equation}
\alpha=\frac{1}{1+\rho^b}\quad\in[0,1],\quad\quad b\geq 0.
\end{equation}
Thus, the transmitted powers of the common and private messages, in the large-$\rho$ scale, can be written as\footnote{Throughout the work, we will use $\dot{=}$ to denote exponential equality, i.e. $f(z)\dot{=}z^b$ means that $\lim_{z\rightarrow\infty}\frac{\log f(z)}{\log z}=b$, $\dot{\leq}$ and $\dot{\geq}$ are defined similarly.}
\begin{equation}
P_{2,\rm{private}}=\frac{\rho}{1+\rho^b}\qquad\text{and}\qquad P_{2,\rm{common}}\doteq\rho.
\end{equation}

We begin by describing the specifics of the transmission scheme under the HK splitting. When the two transmitters receive a NACK at round $l$, they both begin the transmission of the next block of their current messages. If, on the other hand, TX1 receives a NACK while TX2 receives an ACK, we stipulate that TX2 stops sending the private portion of its message and keeps sending the common part until TX1 receives an ACK. We motivate this transmission policy by observing two special cases of the HK-splitting. The first special case is when TX2 uses the CMO scheme \cite{me}. In this case, it is intuitive that the best that TX2 can do when receiving an ACK while TX1 receives a NACK is to keep sending the same message until TX1 receives an ACK. The reason is that RX1 then performs joint decoding for {\it both} messages from TX1 and TX2. Thus, when TX2 keeps sending the same message, RX1 will accumulate more joint mutual information. Hence reducing the probability of the joint outage event at RX1.

The other special scenario is the TIAN scheme which can be obtained directly from the two-message HK approach by setting $b=0$ and $t_2=0$ \cite{me}. Under the TIAN scheme, and contrary to the CMO counterpart, we expect the diversity at RX1 to improve if TX2 ceases the transmission of its current message when receiving an ACK while TX1 receives a NACK since this provides for less interference. Note that this will not affect the diversity at RX2. The HK scheme with generic splitting parameters lies midways between those two special schemes and it is for this reason that we stipulate the stopping of the private message when a NACK is received at RX1. It is noteworthy that the average transmitted power at TX1 or TX2 will not be affected by either continuing or stopping the transmission of the same message after receiving an ACK and until the other transmitter receives an ACK as the probabilities of such events are very small for the case of the large-$\rho$ scale.

We demonstrated in \cite{me} that the CMO scheme is a singular special case of the HK approach. So, we now state the achievable three dimensional tradeoff of the non-cooperative ARQ protocol under the use of the HK and the CMO approaches as they are distinct.

\begin{theorem}
The Achievable diversity, multiplexing, and delay tradeoff of a two user Rayleigh fading ZIC under the use of the non-cooperative ARQ protocol with a maximum of $L$ transmission rounds for the HK approach and using our transmission policy is
\begin{equation}
\begin{split}
&d_{1,\rm{HK}}(L)=\\
&\underset{i\in\left\{1,2,\cdots,L\right\}}\min\Bigg\{\min\left\{\left[1-\frac{r_2}{i-1}\right]^+,\left[1-\frac{r_2-t_2}{i-1}-b\right]^+\right\}\\
&\qquad\qquad\qquad\qquad\;+\min\left\{d_{11,\rm{HK}}(L,i),d_{12,\rm{HK}}(L,i)\right\}\Bigg\},\\
&\text{where,}\\
&d_{11,\rm{HK}}(L,i)=\\
&\qquad\;\;\max\left\{\left[1-\frac{r_1}{L-i}\right]^+,\left[1-\frac{r_1+i\left[\beta-b\right]^+}{L}\right]^+\right\}\\
&d_{12,\rm{HK}}(L,i)=\\
&\begin{cases}
\left[1-\frac{(r_1+t_2)+i\left[\beta-b\right]^+}{L}\right]^+,\;\text{if}\; r_1+t_2\geq(L-i)\beta+ib>Lb\\
\left[1-\frac{(r_1+t_2)-ib}{L-i}\right]^++\left[\beta-\frac{(r_1+t_2)-ib}{L-i}\right]^+,\\
\qquad\qquad\qquad\qquad\text{if}\;\;Lb<r_1+t_2<(L-i)\beta+ib\\
\left[1-\frac{r_1+t_2}{L}\right]^++\left[\beta-\frac{r_1+t_2}{L}\right]^+,\;\;\text{if}\;\;r_1+t_2\leq Lb.
\end{cases}\\
&\text{And,}\\
&d_{2,\rm{HK}}(L)=\min\left\{\left[1-\frac{r_2}{L}\right]^+,\left[1-\frac{r_2-t_2}{L}-b\right]^+\right\}.
\end{split}
\label{eq:HKDMT}
\end{equation}

While the achievable tradeoff under the CMO scheme is given by
\begin{equation}
\begin{split}
&d_{1,\rm{CMO}}(L)=\\
&\min\left\{\left[1-\frac{r_1}{L}\right]^+,\left[1-\frac{r_1+r_2}{L}\right]^++\left[\beta-\frac{r_1+r_2}{L}\right]^+\right\}
\end{split}
\label{eq:CMO1DMT}
\end{equation}
\begin{equation}
d_{2,\rm{CMO}}(L)=\left[1-\frac{r_2}{L}\right]^+.
\label{eq:CMO2DMT}
\end{equation}

\end{theorem}

\begin{proof}
Following in the footsteps of \cite{ArqHesham}, it is immediate to show that the individual error probabilities are exponentially equal to their respective outage probabilities for sufficiently large $T$. This can be qualitatively illustrated as follows. The use of joint typical-set decoding limits the probability of the undetected error event at any round $l\leq L$ to an arbitrarily small value. Following the same techniques in \cite{TseDiv}, it can be directly shown that the probability of decoding failure at round $l=L$ at either RX1 or RX2 is exponentially equal to the probability of the corresponding outage event at the end of the $L$ transmission rounds. Thus, we have for $i=1,2$

\begin{equation}
\begin{split}
P_{e_i}(L,\rho)&\doteq\rho^{-d_{i,\rm{HK}}(L)}\\
&\doteq P_{\rm{out,i}}(L,\rho)\\
&\doteq\rho^{-d_{\rm{out,i}}(L)},
\end{split}
\label{eq:erroroutage}
\end{equation}
where $P_{\rm{out,i}}(L)$ is the individual outage probability at RXi. Note that $d_{\rm{out,i}}(L)$ denotes the diversity gain associated with $P_{\rm{out,i}}(L)$.

We then derive the individual outage probabilities for the non-cooperative ARQ-ZIC system. When the accumulated mutual information over the consecutive rounds at RX1(RX2) is smaller than the first block rate $R_1(R_2)$, an outage occurs. It was shown in \cite{ArqHesham} that it is sufficient, without loss of optimality, to assume that the input codewords are Gaussian distributed. Thus, the mutual information is identical over the protocol rounds. Let us redefine $\overline{\cal{A}}_l$ and $\overline{\cal{B}}_l$ as the outage events at RX1 and RX2 at round $l$, respectively. For the HK approach, the outage region at RX2 at round $l$ can be written as
\begin{equation}
\begin{split}
\overline{\cal{B}}_l=\Bigg\{h_{22}:\;&l\log\left(1+|h_{22}|^2\rho\right)<{R_2},\\
&\text{or}\;\;\;l\log\left(1+|h_{22}|^2\rho\right)<{T_2},\\
&\text{or}\;\;\;l\log\left(1+\frac{|h_{22}|^2\rho}{1+\rho^b}\right)<{R_2-T_2}\Bigg\}.
\end{split}
\label{eq:outage2}
\end{equation}

Notice that the outage event $l\left(1+|h_{22}|^2\rho\right)<{T_2}$ is a subset of the outage event $l\log\left(1+|h_{22}|^2\rho\right)<{R_2}$. Hence, it can be eliminated. Therefore, the high-$\rho$ approximation of the outage region at RX2 at round $l=L$ can be given by
\begin{equation}
\begin{split}
\overline{\cal{B}}_L=\left\{\gamma_{22}:\;L\left[1-\gamma_{22}\right]^+<r_2,\;\;\text{or}\;\;L\left[1-\gamma_{22}-b\right]^+<r_2-t_2\right\}.
\end{split}
\label{eq:outage2-HSNR}
\end{equation}

Following similar analysis as in \cite{me}, the outage probability at RX2 at round $l=L$, $P_{\rm{out,2}}(L)={\rm{Pr}}(\overline{\cal{B}}_L)$, can be shown to be as follows.
\begin{equation}
\begin{split}
P_{\rm{out,2}}(L)&\doteq\rho^{-\min\left\{\left[1-\frac{r_2}{L}\right]^+,\left[1-\frac{r_2-t_2}{L}-b\right]^+\right\}}\\
&\doteq\rho^{-d_{2,\rm{HK}}(L)}.
\end{split}
\label{eq:O2}
\end{equation}
Thus, we have
\begin{equation}
d_{2,\rm{HK}}(L)=\min\left\{\left[1-\frac{r_2}{L}\right]^+,\left[1-\frac{r_2-t_2}{L}-b\right]^+\right\}.
\label{eq:d2_HK_L}
\end{equation}

We define ${\cal{C}}_i$ as the event that TX2 receives an ACK at round $i$ and receives a NACK at round $i-1$, thus, ${\cal{C}}_i=\left\{\overline{\cal{B}}_{i-1},{\cal{B}}_i\right\}$. Notice that a NACK at round $i-1$ implies a NACK at every round $l<i-1$. The outage region at RX1 given ${\cal{C}}_i$ at round $l$ can be written as
\begin{equation}
\begin{split}
\overline{\cal{A}}_l|{\cal{C}}_i&= \Bigg\{h_{11},h_{21}:\;i\log\left(1+\frac{|h_{11}|^2\rho}{1+\frac{|h_{21}|^2\rho^\beta}{1+\rho^b}}\right)\\
&\qquad\qquad+(l-i)\log\left(1+|h_{11}|^2\rho\right)<R_1\\
&\text{or}\;\; i\log\left(1+\frac{|h_{11}|^2\rho+|h_{21}|^2\rho^\beta}{1+\frac{|h_{21}|^2\rho^\beta}{1+\rho^b}}\right)+\\
&(l-i)\log\left(1+|h_{11}|^2\rho+|h_{21}|^2\rho^\beta\right)<R_1+T_2\Bigg\}.
\end{split}
\label{eq:outage1}
\end{equation}

The outage probability at RX1 at round $l=L$ can be derived as follows.
\begin{equation}
\begin{split}
P_{\rm{out,1}}(L)&=\sum_{i=1}^{L}{\rm{Pr}}(\overline{\cal{A}}_L|{\cal{C}}_i){\rm{Pr}}({\cal{C}}_i)\\
&\doteq\rho^{-d_{1,\rm{HK}}(L)}.
\end{split}
\label{eq:O1}
\end{equation}
Using the outage events given in (\ref{eq:outage1}), we show in the Appendix that
\begin{equation}
{\rm{Pr}}(\overline{\cal{A}}_L|{\cal{C}}_i)\doteq\rho^{-\min\left\{d_{11,\rm{HK}}(L,i),d_{12,\rm{HK}}(L,i)\right\}},
\label{ai}
\end{equation}
where,
\begin{equation}
\begin{split}
&d_{11,\rm{HK}}(L,i)=\\
&\max\left\{\left[1-\frac{r_1}{L-i}\right]^+,\left[1-\frac{r_1+i\left[\beta-b\right]^+}{L}\right]^+\right\}\\
&d_{12,\rm{HK}}(L,i)=\\
&\begin{cases}
\left[1-\frac{(r_1+t_2)+i\left[\beta-b\right]^+}{L}\right]^+,\;\text{if}\;r_1+t_2\geq(L-i)\beta+ib>Lb\\
\left[1-\frac{(r_1+t_2)-ib}{L-i}\right]^++\left[\beta-\frac{(r_1+t_2)-ib}{L-i}\right]^+,\\
\qquad\qquad\qquad\qquad\qquad\text{if}\;\; Lb<r_1+t_2<(L-i)\beta+ib\\
\left[1-\frac{r_1+t_2}{L}\right]^++\left[\beta-\frac{r_1+t_2}{L}\right]^+,\qquad\text{if}\;\;r_1+t_2\leq Lb.
\end{cases}
\end{split}
\label{eq:d_11,d_21(L,i)}
\end{equation}

The probability of the event ${\cal{C}}_i$ can be derived as follows.
\begin{equation}
\begin{split}
{\rm{Pr}}({\cal{C}}_i)&={\rm{Pr}}(\overline{\cal{B}}_{i-1},{\cal{B}}_i)\\
&={\rm{Pr}}(\overline{\cal{B}}_{i-1}){\rm{Pr}}({\cal{B}}_i|\overline{\cal{B}}_{i-1})\\
&\doteq {\rm{Pr}}(\overline{\cal{B}}_{i-1})\\
&\doteq\rho^{-d_{2,\rm{HK}}(i-1)},
\end{split}
\label{ci}
\end{equation}
where, ${\rm{Pr}}({\cal{B}}_i|\overline{\cal{B}}_{i-1})\doteq1$. Thus, using (\ref{ai}) and (\ref{ci}) in (\ref{eq:O1}), we get
\begin{equation}
P_{out,1}(l)\doteq\sum_{i=1}^{L}\rho^{-\left\{d_{2,\rm{HK}}(i-1)+\min\{d_{11,\rm{HK}}(l,i),d_{12,\rm{HK}}(l,i)\}\right\}}.
\label{eq:P_out_1}
\end{equation}

In the high-$\rho$ scale, the term that dominates the previous summation is the one with the minimum negative exponent. Using (\ref{eq:O1}) and (\ref{eq:P_out_1}), we have
\begin{equation}
\begin{split}
d_{1,\rm{HK}}(L)=&\\
\underset{i\in\left\{1,2,\cdots,L\right\}}\min&\Bigg\{\min\left\{\left[1-\frac{r_2}{i-1}\right]^+,\left[1-\frac{r_2-t_2}{i-1}-b\right]^+\right\}\\
&+\min\left\{d_{11,\rm{HK}}(L,i),d_{12,\rm{HK}}(L,i)\right\}\Bigg\},
\end{split}
\end{equation}
where $d_{11,\rm{HK}}(L,i)$ and $d_{12,\rm{HK}}(L,i)$ are as given in (\ref{eq:d_11,d_21(L,i)}).

Now for the CMO scheme, the outage regions at RX1 and RX2 at round $l$ can be given as follows.
\begin{equation}
\begin{split}
\overline{\cal{A}}_l=\Bigg\{&h_{11},h_{21}:l\log\left(1+|h_{11}|^2\rho\right)<R_1,\\
&\text{or}\;\;\;l\log\left(1+|h_{11}|^2\rho+|h_{21}|^2\rho^{\beta}\right)<R_1+R_2\Bigg\}
\end{split}
\label{eq:outage1CMO}
\end{equation}
\begin{equation}
\overline{\cal{B}}_l=\left\{h_{22}:l\log\left(1+|h_{22}|^2\rho\right)<R_2\right\}.
\label{eq:outage2CMO}
\end{equation}

The high-$\rho$ approximation of these outage regions at round $l=L$ can be given by
\begin{equation}
\begin{split}
\overline{\cal{A}}_L=\Bigg\{&\gamma_{11},\gamma_{21}:L\left[1-\gamma_{11}\right]^+<r_1,\\
&\text{or}\;\;L\left[\max\left\{\left[1-\gamma_{11}\right]^+,\left[\beta-\gamma_{21}\right]^+\right\}\right]<r_1+r_2\Bigg\}
\end{split}
\label{eq:outage1CMO-HSNR}
\end{equation}
\begin{equation}
\overline{\cal{B}}_L=\left\{\gamma_{22}:L\left[1-\gamma_{22}\right]^+<r_2\right\}.
\label{eq:outage2CMO-HSNR}
\end{equation}

Using these outage regions and following similar steps \cite{me}, we can easily show that the individual diversities of the CMO non-cooperative ARQ ZIC setting with maximum of $L$ transmission rounds can be given by
\begin{equation}
\begin{split}
&d_{1,\rm{CMO}}(L)=\\
&\min\left\{\left[1-\frac{r_1}{L}\right]^+,\left[1-\frac{r_1+r_2}{L}\right]^++\left[\beta-\frac{r_1+r_2}{L}\right]^+\right\}
\end{split}
\end{equation}
\begin{equation}
d_{2,\rm{CMO}}(L)=\left[1-\frac{r_2}{L}\right]^+.
\end{equation}

To complete the achievability proof, we have to show that the effective throughputs $\eta_1$ and $\eta_2$ are exponentially equal to their corresponding first block rates, $R_1$ and $R_2$. Thus, the effective multiplexing gains $r_{e_1}$ and $r_{e_2}$ are equal to the first block multiplexing gains $r_1$ and $r_2$, respectively. Using equations (\ref{eq:throughput}) and (\ref{eq:erroroutage}), and recalling the definitions for $\overline{\cal{A}}_l$ and $\overline{\cal{B}}_l$, we have
\begin{equation}
\begin{split}
\eta_1&\dot{\geq}\frac{R_1}{1+\sum_{l=1}^{L-1}\rho^{-\min\left\{d_{\rm{out,1}}(l),d_{\rm{out,2}}(l)\right\}}}\\
&\dot{=}R_1\\
\eta_2&\dot{\geq}\frac{R_2}{1+\sum_{l=1}^{L-1}\rho^{-\min\left\{d_{\rm{out,1}}(l),d_{\rm{out,2}}(l)\right\}}}\\
&\dot{=}R_2.
\end{split}
\label{eq:th}
\end{equation}
Now, from equations (\ref{eq:throughput11}) and (\ref{eq:th}), we directly have
\begin{equation}
\eta_1\doteq R_1\qquad\text{and}\qquad\eta_2\doteq R_2,
\end{equation}
which yields that
\begin{equation}
r_{e_1}\doteq r_1\qquad\text{and}\qquad r_{e_2}\doteq r_2,
\end{equation}
\end{proof}

We can show the superiority of our transmission policy over other approaches which consider keeping or stopping the transmission of {\it{both}} the common and private messages of TX2 when it receives an ACK while TX1 receives a NACK as follows.

For our transmission policy under no cooperation and using the HK approach, the rate region at RX1 at round $l=L$, $R_{1,\rm{HK-NC}}$, can be expressed as
\begin{equation}
R_{1,\rm{HK-NC}}(L)= \bigcup_{i=1}^{L}\left\{{\cal{A}}_L|{\cal{C}}_i\right\},
\end{equation}
where, ${\cal{C}}_i$ is as previously defined. Also, $\left\{{\cal{A}}_L|{\cal{C}}_i\right\}$ is the rate region at RX1 given ${\cal{C}}_i$ at round $l=L$ which can be written as follows.
\begin{equation}
\begin{split}
&{\cal{A}}_L|{\cal{C}}_i=\\
&\Bigg\{h_{11},h_{21}:\;R_1\leq i\log\left(1+\frac{|h_{11}|^2\rho}{1+\frac{|h_{21}|^2\rho^\beta}{1+\rho^b}}\right)\\
&\qquad\qquad\qquad\qquad+(L-i)\log\left(1+|h_{11}|^2\rho\right),\\
&\qquad\text{and}\;\;R_1+R_2\leq i\log\left(1+\frac{|h_{11}|^2\rho+|h_{21}|^2\rho^\beta}{1+\frac{|h_{21}|^2\rho^\beta}{1+\rho^b}}\right)\\
&\qquad\qquad\qquad\qquad+(L-i)\log\left(1+|h_{11}|^2\rho+|h_{21}|^2\rho^\beta\right)\Bigg\}.
\end{split}
\label{eq:rateregion-OP}
\end{equation}

The rate region at RX1 of the approach which considers continuing the transmission of both the common and private message when TX2 receives an ACK while TX1 receives a NACK can be obtained from $\left\{{\cal{A}}_L|{\cal{C}}_i\right\}$ by setting $i=L$. Therefore, the rate region of this approach is a subset of the rate region of our transmission policy; this ultimately shows the superiority of our transmission policy.

On the other hand, the rate region at RX1 of the approach which considers stopping the transmission of both the common and private message when TX2 receives an ACK while TX1 receives a NACK is obtained from $R_{1,\rm{HK-NC}}(L)$ by removing the common message power $|h_{21}|^2\rho^\beta$ from the second term in the right hand side of the second constraint in (\ref{eq:rateregion-OP}). Thus, it is also a subset of the rate region of our transmission policy.

Notice that the rate regions at RX2 for these two approaches are similar to that of our transmission policy, which can be defined as the complement of the outage region at RX2 at round $l=L$ given in (\ref{eq:outage2-HSNR}).

For the TIAN scheme, substituting with $b=0$ and $t_2=0$ in the equations given in (\ref{eq:HKDMT}) yields
\begin{equation}
\begin{split}
&d_{1,\rm{TIAN}}(L)=\underset{i\in\{1,2,\cdots,L\}}\min\\
&\Bigg\{\left[1-\frac{r_2}{i-1}\right]^++\max\left\{\left[1-\frac{r_1}{L-i}\right]^+,\left[1-\frac{r_1+i\beta}{L}\right]^+\right\}\Bigg\} \end{split}
\end{equation}
\begin{equation}
d_{2,\rm{TIAN}}(L)=\left[1-\frac{r_2}{L}\right]^+.
\label{eq:TIAN2DMT}
\end{equation}
We can show that $i=1$ minimizes the expression for $d_{1,\rm{TIAN}}(L)$ in the above equation, thus, we have
\begin{equation}
d_{1,\rm{TIAN}}(L)=\max\left\{\left[1-\frac{r_1}{L-1}\right]^+,\left[1-\frac{r_1}{L}-\frac{\beta}{L}\right]^+\right\}.
\label{eq:TIAN1DMT}
\end{equation}

\subsection{The Cooperative ARQ Protocol}\label{coop}
We investigate here the achievable tradeoff of two cooperative ARQ schemes. In both schemes, TX2, the interfering link in the ZIC model, assists in relaying the message of TX1 in the event of a NACK reception at TX1. This setting can model a coexistence scenario between a primary link and a secondary link in a cognitive radio setting. The goal of TX2 is to access the wireless medium while preserving the primary transmitter's, TX1, privileged access. We show that the cooperative ARQ schemes significantly improve the diversity of the primary link.

We first consider a static decoding scheme where the decoding scheme at RX1, whether using the CMO or the TIAN decoding, is fixed and determined a priori according to the interference level $\beta$ and the multiplexing gains $r_1$ and $r_2$. Next, we consider a dynamic decoding scheme where RX1 dynamically decides at the beginning of {\it each} new transmission, when the channel gains change, to use either the CMO or the TIAN form of decoding according to the channel gains, the interference level, and the multiplexing gains. 

For the two cooperative ARQ schemes, if TX1 receives a NACK for its message, TX2 will start listening to TX1 to decode its message regardless of its own feedback. We denote the time TX2 takes to decode TX1 message by $T'$. We expect the ZIC-system in those Cooperative ARQ schemes with $L$ ARQ rounds to act as a $2\times1$ MISO ARQ system as the number of retransmission rounds increases. This is evident for the diversity results we report here considering $L=2$ transmission rounds only.

Since $T'$ is the time TX2 takes to accumulate enough mutual information to decode the message from TX1 with first block rate $R_1$, the following relation holds
\begin{equation}
T'=\min\left\{T,\left\lceil \frac{TR_1}{\log_2(1+|h|^2\rho)}\right\rceil\right\},
\label{eq:threshold}
\end{equation}
\noindent where $h$ is the channel gain between TX1 and TX2.

Once TX2 has decoded TX1's message, it will start relaying this message using a codebook $\tilde{C}_1(\rho)$. We denote the codeword used by TX2 to encode the message of TX1 by $\bold{\tilde{x}}_{1,3}$. If TX2 decodes the primary message in $T'$ symbols, then it will assist TX1 by relaying its message in the remaining time of the second transmission round. This means that $\bold{\tilde{x}}_{1,3}$ is a complex vector of length $T-T'$. Based on this communication setup, the received signal at RX1, in the event of a NACK reception by TX1 at the end of the first transmission round, takes the following form depending on the transmission phase.

\begin{enumerate}
\item \textbf{The first transmission phase}\\
In the first transmission round, the received signal at RX1 can be written as
\begin{itemize}
\item For the CMO decoder
\begin{equation}
\begin{split}
&\bold{y}_{1,1}=h_{11}\bold{x}_{1,1}+h_{21}\bold{x}_{2,1}+\bold{n}_{1,1},\\
&\text{where,}\\
&\bold{n}_{1,1}\sim{\mathcal{CN}}(0,I_T).
\label{eq:first_round}
\end{split}
\end{equation}
\item For the TIAN decoder
\begin{equation}
\begin{split}
&\bold{y}_{1,1}=h_{11}\bold{x}_{1,1}+\bold{n'}_{1,1},\\
&\text{where,}\\
&\bold{n'}_{1,1}\sim{\mathcal{CN}}\left(0,I_T\left(1+|h_{21}|^2\rho^{\beta}\right)\right).
\end{split}
\label{eq:first_roundt}
\end{equation}
\end{itemize}

\item \textbf{The listening phase}\\
If a NACK is received by TX1 at the end of the first transmission round, the received signal at RX1 during the listening phase can be written as
\begin{equation}
\begin{split}
&\bold{y}_{1,2}=h_{11}\bold{x}_{1,2}+\bold{n}_{1,2},\\
&\text{where,}\\
&\bold{n}_{1,2}\sim{\mathcal{CN}}(0,I_{T'}).
\end{split}
\label{eq:listen_round}
\end{equation}

\item \textbf{The Cooperation phase}\\
If TX2 decodes the message from TX1 in a time $T'<T$, the received signal at RX1 during the cooperation phase can be written as
\begin{equation}
\begin{split}
&\bold{y}_{1,3}=h_{11}\bold{x}_{1,3}+h_{21}\bold{\tilde{x}}_{1,3}+\bold{n}_{1,3},\\
&\text{where,}\\
&\bold{n}_{1,3}\sim{\mathcal{CN}}(0,I_{T-T'}).
\end{split}
\label{eq:help_round}
\end{equation}
\end{enumerate}

\subsubsection{Cooperative ARQ with Static Decoding}
We characterize here the achievable DMT of the cooperative ARQ with static decoding considering a maximum of two transmission rounds.\footnote{Now that we have fixed delay at two rounds, we focus on the resulting DMT.} Herein, we restrict ourselves to the use of the CMO or the TIAN decoding for simplicity. We use the superscript $^c$ to refer to the cooperation setup. Our results are detailed in the following theorem.

\begin{theorem}
The achievable DMT of the cooperative ARQ with static decoding scheme with maximum of two transmission rounds and under the use of the CMO scheme can be characterized as follows.
\begin{equation}
\begin{split}
&d_{1,\rm{CMO}}^c(2)=\min\left\{d_{11,\rm{CMO}}^c(2),d_{12,\rm{CMO}}^c(2)\right\}\\
&\text{where,}\\
&d_{11,\rm{CMO}}^c(2)=\\
&\begin{cases}
1-\frac{r_{1}}{2},\qquad\text{if}\;\;r_1\geq2\beta\\
\min\left\{1+\frac{(1-r_{1})\beta-r_{1}}{1+r_{1}},2-\frac{3r_{1}}{2}\right\},\;\; \text{if}\;\;\frac{\beta}{1+\beta}\leq r_1<2\beta\\
\min\left\{2-\frac{3r_{1}}{2},2-\frac{\beta r_1}{\beta-r_1},1+\beta-\frac{r_1}{1-r_1}\right\},\;\; \text{if}\;\;r_1<\frac{\beta}{1+\beta}
\end{cases}\\
&d_{12,\rm{CMO}}^c(2)=\left[1-\frac{r_1+r_2}{2}\right]^++\left[\beta_-\frac{r_1+r_2}{2}\right]^+.\\
&\text{And,}\\
&d_{2,\rm{CMO}}^c(2)=\min\left\{d_{21,\rm{CMO}}^c(2),d_{22,\rm{CMO}}^c(2)\right\},\\
&\text{where,}\\
&d_{21,\rm{CMO}}^c(2)=\\
&\min\left\{\left[1-r_1\right]^+,\left[1-r_1-r_2\right]^++\left[\beta-r_1-r_2\right]^+\right\}+\left[1-r_2\right]^+\\
&d_{22,\rm{CMO}}^c(2)=\left[1-\frac{r_2}{2}\right]^+.
\end{split}
\label{eq:COPSTATCMO}
\end{equation}
While for the TIAN scheme, the achievable DMT can be expressed as follows.
\begin{equation}
\begin{split}
&d_{1,\rm{TIAN}}^c(2)=\begin{cases}
&\left[1-\frac{r_1+\beta}{2}\right]^+,\qquad\text{if}\;\;r_1\geq\beta\\
&2\left[1-r_1\right]^+,\qquad\text{if}\;\;r_1<\frac{\beta}{2},\;\beta\geq1\\
&\left[1-r_1\right]^++\left[\beta-r_1\right]^+,\qquad\text{if}\;\;r_1<\frac{\beta}{2},\;\beta<1\\
&\frac{\left(1-r_1\right)\beta}{r_1},\qquad\text{if}\;\;r_1>\frac{1}{2},\;\frac{\beta}{2}\leq r_1<\beta\\
&\left[1-r_1\right]^++\left[\beta-r_1\right]^+,\;\;\text{if}\;\;r_1\leq\frac{1}{2},\;\frac{\beta}{2}\leq r_1<\beta\\
\end{cases}\\
&d_{2,\rm{TIAN}}^c(2)=\min\left\{d_{21,\rm{TIAN}}^c(2),d_{22,\rm{TIAN}}^c(2)\right\}\\
&\text{where,}\\
&d_{21,\rm{TIAN}}^c(2)=\left[1-r_2\right]^++\left[1-r_1-\beta\right]^+\\
&d_{22,\rm{TIAN}}^c(2)=\left[1-\frac{r_2}{2}\right]^+.
\end{split}
\label{eq:COPSTATTIAN}
\end{equation}
The overall achievable DMT curve, either between RX1 diversity and first user multiplexing gain $r_1$ or between RX2 diversity and second user multiplexing gain $r_2$, of the cooperative ARQ with static decoding scheme for $L=2$ is the maximum of the achievable DMT using the CMO and the TIAN approaches.
\end{theorem}

\begin{proof}
For the cooperative ARQ with static decoding, error at RX1 ${\cal{E}}_1$ is comprised of the following events.
\begin{enumerate}
\item $\left\{{\cal{E}}_{1,{\cal{E}}_{12}}\right\}$ denotes the error event at RX1 when TX2 makes an error in decoding the first user message.
\item $\left\{{\cal{E}}_{1,\overline{\cal{E}}_{12}}\right\}$ denotes the error event at RX1 when TX2 decodes the first user message correctly. This event can be expressed as the union of the two following events.
    \begin{itemize}
    \item $\left\{{\cal{E}}_1,{\cal{A}}_{1}\right\}$ denotes the event of an undetected decoding error at RX1 at the end of round  1.
    \item $\left\{{\cal{E}}_1,\overline{\cal{A}}_1\right\}$ denotes the event of a decoding failure at RX1 at the end of the first transmission round. It can be written as the union of the two following events. A decoding failure $\left\{{\cal{E}}_1,\overline{\cal{A}}_1,\overline{\cal{A}}_2\right\}$ and an undetected decoding error $\left\{{\cal{E}}_1,\overline{\cal{A}}_1,{\cal{A}}_2\right\}$ at RX1 at the end of the second transmission round.
    \end{itemize}
\end{enumerate}
Therefore, the error at RX1 occurs due to the events $\left\{{\cal{E}}_{1,{\cal{E}}_{12}}\right\}$, $\left\{{\cal{E}}_1,{\cal{A}}_{1}\right\}$, $\left\{{\cal{E}}_1,\overline{\cal{A}}_1,{\cal{A}}_2\right\}$, and $\left\{{\cal{E}}_1,\overline{\cal{A}}_1,\overline{\cal{A}}_2\right\}$.

Similar to the work in \cite{ArqHesham} and \cite{Kambiz}, it can be shown that for a sufficiently large block length $T$ the events $\left\{{\cal{E}}_{1,{\cal{E}}_{12}}\right\}$, $\left\{{\cal{E}}_1,{\cal{A}}_{1}\right\}$, and $\left\{{\cal{E}}_1,\overline{\cal{A}}_1,{\cal{A}}_2\right\}$ can be made arbitrary small. Thus, the error event at RX1 ${\cal{E}}_1$ is dominated by the event $\left\{{\cal{E}}_1,\overline{\cal{A}}_1,\overline{\cal{A}}_2\right\}$ which corresponds to an outage event at RX1 at the end of the second transmission round $\left\{\overline{\cal{A}}_2\right\}$. Therefore, the error probability at RX1 is exponentially equal to the probability of the outage event $\left\{\overline{\cal{A}}_2\right\}$.

We now derive the outage probability at RX1 at the end of the second transmission round for the CMO scheme. Let us state the corresponding outage event as follows.
\begin{equation}
\overline{\cal{A}}_2=\left\{{\cal{F}}_{T'},\left\{{\cal{O}}_1\cup{\cal{O}}_2\right\}\right\},
\end{equation}
where,
\begin{equation}
{\cal{F}}_{T'}=\left\{h: \frac{T'}{T}\log\left(1+|h|^2\rho\right)=R_1\right\}
\label{eq:COP-STAT-F-CMO}
\end{equation}
\begin{equation}
\begin{split}
{\cal{O}}_1=&\Bigg\{h_{11},h_{21}:\;\frac{T+T'}{T}\log\left(1+|h_{11}|^2\rho\right)+\\
&\frac{T-T'}{T}\log\left(1+|h_{11}|^2\rho+|h_{21}|^2\rho^{\beta}\right)<R_1\Bigg\}
\end{split}
\label{eq:COP-STAT-OUT1-CMO}
\end{equation}
\begin{equation}
\begin{split}
&{\cal{O}}_2=\Bigg\{h_{11},h_{21}:\\
&\frac{T}{T}\log\left(1+|h_{11}|^2\rho+|h_{21}|^2\rho^{\beta}\right)+\frac{T'}{T}\log\left(1+|h_{11}|^2\rho\right)\\
&+\frac{T-T'}{T}\log\left(1+|h_{11}|^2\rho+|h_{21}|^2\rho^{\beta}\right)<R_1+R_2\Bigg\}.
\end{split}
\label{eq:COP-STAT-OUT2-CMO}
\end{equation}

Defining $f=\frac{T'}{T}$, $|h_{ij}|^2=\rho^{-\gamma_{ij}}$, and $|h|^2=\rho^{-u}$, the high-$\rho$ approximation of the previous outage events can be written as
\begin{equation}
{\cal{F}}_{T'}=\left\{u: f\left[1-u\right]^+=r_1\right\}
\label{eq:COP-STAT-F-CMO-HSNR}
\end{equation}
\begin{equation}
\begin{split}
{\cal{O}}_1=\Bigg\{&\gamma_{11},\gamma_{21},f:\;\left(1+f\right)\left[1-\gamma_{11}\right]^+\\
&+\left(1-f\right)\max\left\{\left[1-\gamma_{11}\right]^+,[\beta-\gamma_{21}]^+\right\}<r_1\Bigg\}
\end{split}
\label{eq:COP-STAT-OUT1-CMO-HSNR}
\end{equation}
\begin{equation}
\begin{split}
{\cal{O}}_2=\Bigg\{&\gamma_{11},\gamma_{21},f:\;\left(2-f\right)\max\left\{\left[1-\gamma_{11}\right]^+,[\beta-\gamma_{21}]^+\right\}\\
&+f\left[1-\gamma_{11}\right]^+<r_1+r_2\Bigg\},
\end{split}
\label{eq:COP-STAT-OUT2-CMO-HSNR}
\end{equation}
where,
\begin{equation}
u=1-\frac{r_1}{f}.
\label{eq:u}
\end{equation}

Since we have $\underset{f\in[r_1,1]}\max{\rm{Pr}}({\cal{O}}_1)\doteq\rho^{-d_{11,\rm{CMO}}^c(2)}$, thus,
\begin{equation}
d_{11,\rm{CMO}}^c(2)=\underset{\gamma_{11},\gamma_{21},u\in{\cal{O}}_1}\min\left\{\gamma_{11}+\gamma_{21}+u\right\}
\end{equation}

The shaded regions in Fig. \ref{fig:proof5-a} show the constraint regions of $\gamma_{11}$ and $\gamma_{21}$ for the cases $r_1\geq2\beta$, $(1-f)\beta\leq r_1<2\beta$, and $r_1<(1-f)\beta$. Thus,
\begin{equation}
\begin{split}
&d_{11,\rm{CMO}}^c(2)=\underset{f\in[r_1,1]}\min\\
&\begin{cases}
2-\frac{r_1}{2}-\frac{r_1}{f},\qquad\text{if}\;\; r_1\geq2\beta\\
2+\frac{(1-f)\beta-r_1}{1+f}-\frac{r_1}{f},\;\;\text{if}\;\;(1-f)\beta\leq r_1<2\beta\\
2+\beta-\frac{r_1}{1-f}-\frac{r_1}{f},\;\;\text{if}\;\;r_1<(1-f)\beta
\end{cases}
\end{split}
\end{equation}
The function $2-\frac{r_1}{2}-\frac{r_1}{f}$ is monotonically increasing in $f$, thus, its minimum is at $f=r_1$. For $r_1<2\beta$, we have
\begin{equation}
\begin{split}
d_{11,\rm{CMO}}^c(2)=\underset{f\in[r_1,1]}\min\begin{cases}
2+\frac{(1-f)\beta-r_1}{1+f}-\frac{r_1}{f},\;\;\text{if}\;\;\;f\geq 1-\frac{r_1}{\beta}\\
2+\beta-\frac{r_1}{1-f}-\frac{r_1}{f},\;\;\text{if}\;\;\; f< 1-\frac{r_1}{\beta}
\end{cases}
\end{split}
\end{equation}
The function $2+\frac{(1-f)\beta-r_1}{1+f}-\frac{r_1}{f}$ is a concave function over $f\in[r_1,1]$, hence, it is minimized at the edges. Thus, for $r_1\geq 1-\frac{r_1}{\beta}$, the function $2+\frac{(1-f)\beta-r_1}{1+f}-\frac{r_1}{f}$ is minimized at $f=r_1$ or $f=1$. On the other hand, when $r_1<1-\frac{r_1}{\beta}$, it is minimized at $f=1-\frac{r_1}{\beta}$ or $f=1$. The function $2+\beta-\frac{r_1}{1-f}-\frac{r_1}{f}$ is also monotonically increasing in $f$ over $f\in[r_1,1]$ and is minimized at $f=r_1$. Notice that the condition $f<1-\frac{r_1}{\beta}$ implies that $r_1<1-\frac{r_1}{\beta}$ since $r_1\leq f\leq 1$. Based on the above arguments we have

\begin{figure}
\centering
\subfigure[$r_1\geq2\beta$]{
   \includegraphics[width=76mm, height = 50mm] {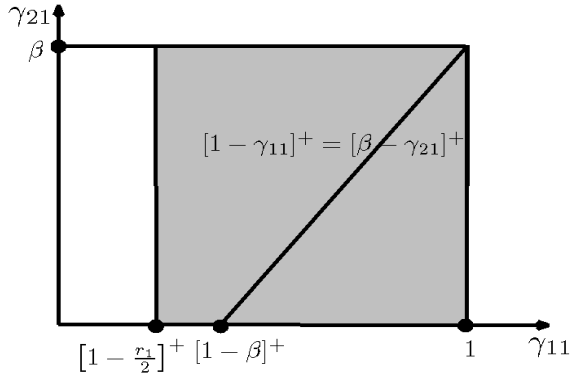}
 }
\subfigure[$(1-f)\beta\geq r_1<2\beta$]{
   \includegraphics[width=76mm, height = 50mm] {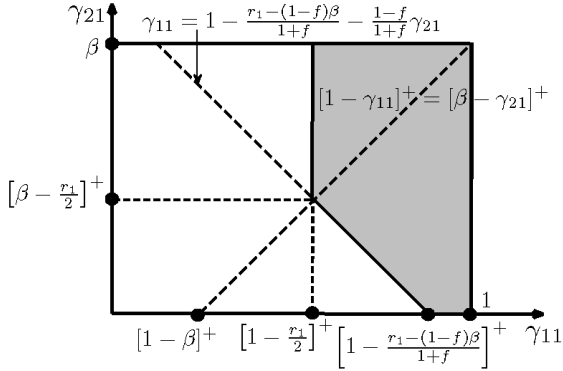}
 }
\subfigure[$r_1<(1-f)\beta$]{
   \includegraphics[width=76mm, height = 50mm] {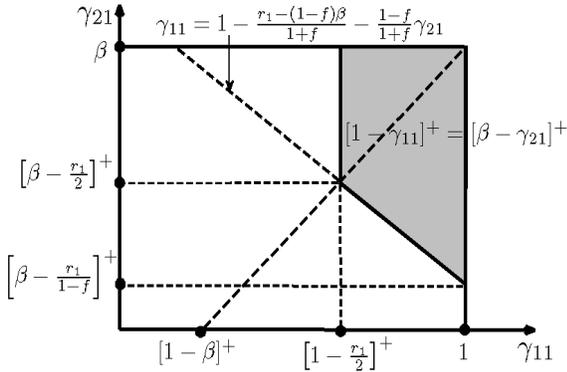}
 }
\caption{Constraint regions of $d_{11,\rm{CMO}}^c(2)$}
\label{fig:proof5-a}
\end{figure}

\begin{equation}
\begin{split}
&d_{11,\rm{CMO}}^c(2)=\\
&\begin{cases}
1-\frac{r_{1}}{2},\qquad\text{if}\;\; r_1\geq2\beta\\
\min\left\{1+\frac{(1-r_{1})\beta-r_{1}}{1+r_{1}},2-\frac{3r_{1}}{2}\right\},\;\;\text{if}\;\;\frac{\beta}{1+\beta}\leq r_1<2\beta\\
\min\left\{2-\frac{3r_{1}}{2},2-\frac{\beta r_1}{\beta-r_1},1+\beta-\frac{r_1}{1-r_1}\right\},\;\; \text{if}\;\;r_1<\frac{\beta}{1+\beta}
\end{cases}
\end{split}
\end{equation}

Similarly, we have $\underset{f\in[r_1,1]}\max{\rm{Pr}}({\cal{O}}_2)\doteq\rho^{-d_{12,\rm{CMO}}^c(2)}$, thus,
\begin{equation}
d_{12,\rm{CMO}}^c(2)=\underset{\gamma_{11},\gamma_{21},u\in{\cal{O}}_2}\min\left\{\gamma_{11}+\gamma_{21}+u\right\}
\end{equation}

The constraint regions of $\gamma_{11}$ and $\gamma_{21}$ for different values of $r_1$ are shown in Fig. \ref{fig:proof5-b}. Thus we have
\begin{equation}
\begin{split}
&d_{12,\rm{CMO}}^c(2)=\underset{f\in[r_1,1]}\min\\
&\begin{cases}
2-\frac{r_1+r_2}{2}-\frac{r_1}{f},\qquad\text{if}\;\;r_1+r_2\geq2\beta\\
\left[1-\frac{r_1+r_2}{2}\right]^++\left[\beta-\frac{r_1+r_2}{2}\right]^++1-\frac{r_1}{f},\;\;\text{if}\;\; r_1+r_2<2\beta.
\end{cases}
\end{split}
\label{eq:d12CMO_f}
\end{equation}
The above two functions are both monotonically increasing in $f$. Therefore, they both are minimized at $f=r_1$. Which yields
\begin{equation}
d_{12,\rm{CMO}}^c(2)=\left[1-\frac{r_1+r_2}{2}\right]^++\left[\beta-\frac{r_1+r_2}{2}\right]^+.
\end{equation}

\begin{figure}
\centering
\subfigure[$r_1+r_2\geq2\beta$]{
   \includegraphics[width=76mm, height = 50mm] {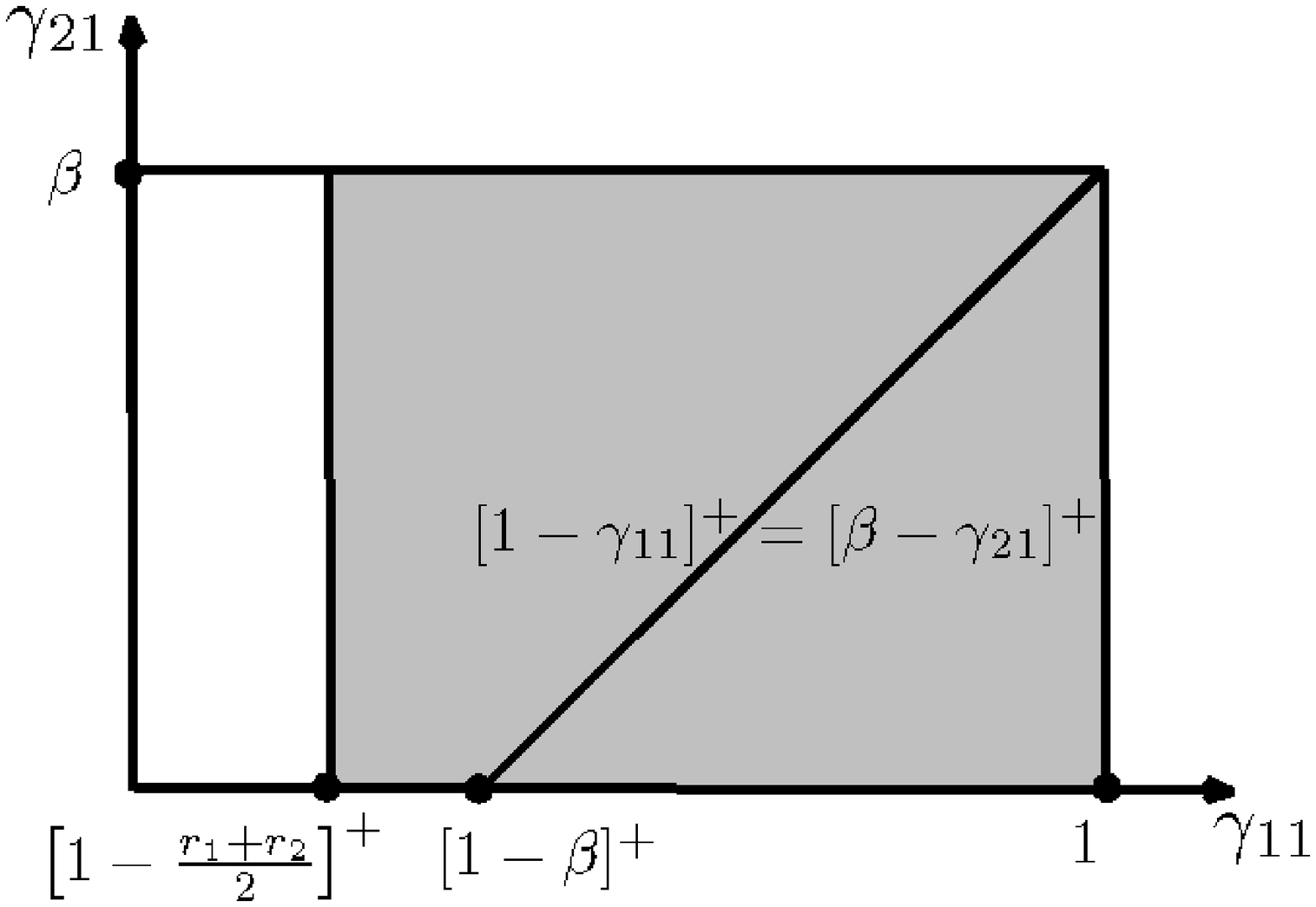}
 }
\subfigure[$(2-f)\beta\leq r_1+r_2<2\beta$]{
   \includegraphics[width=76mm, height = 50mm] {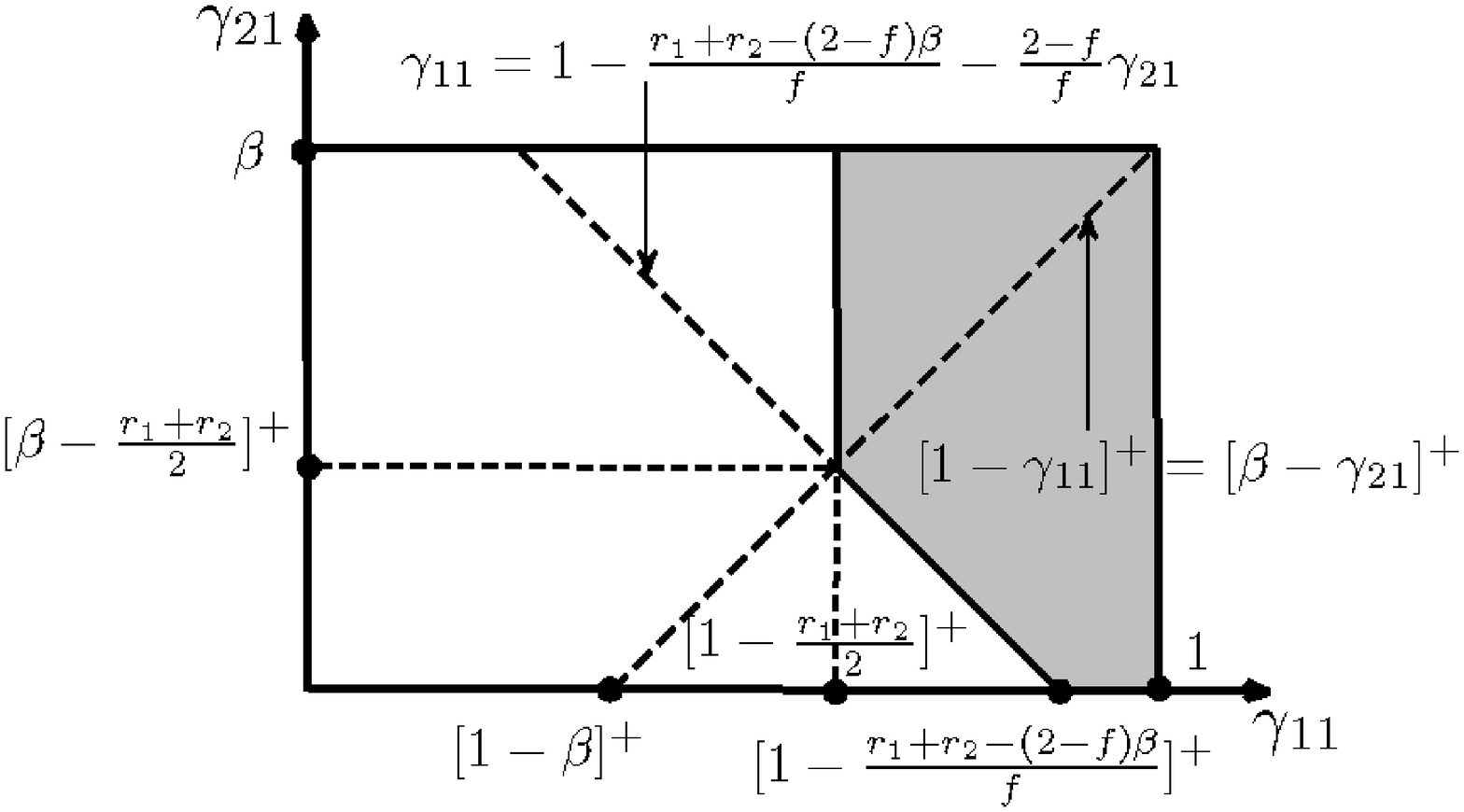}
 }
\subfigure[$r_1+r_2<(2-f)\beta$]{
   \includegraphics[width=76mm, height = 50mm] {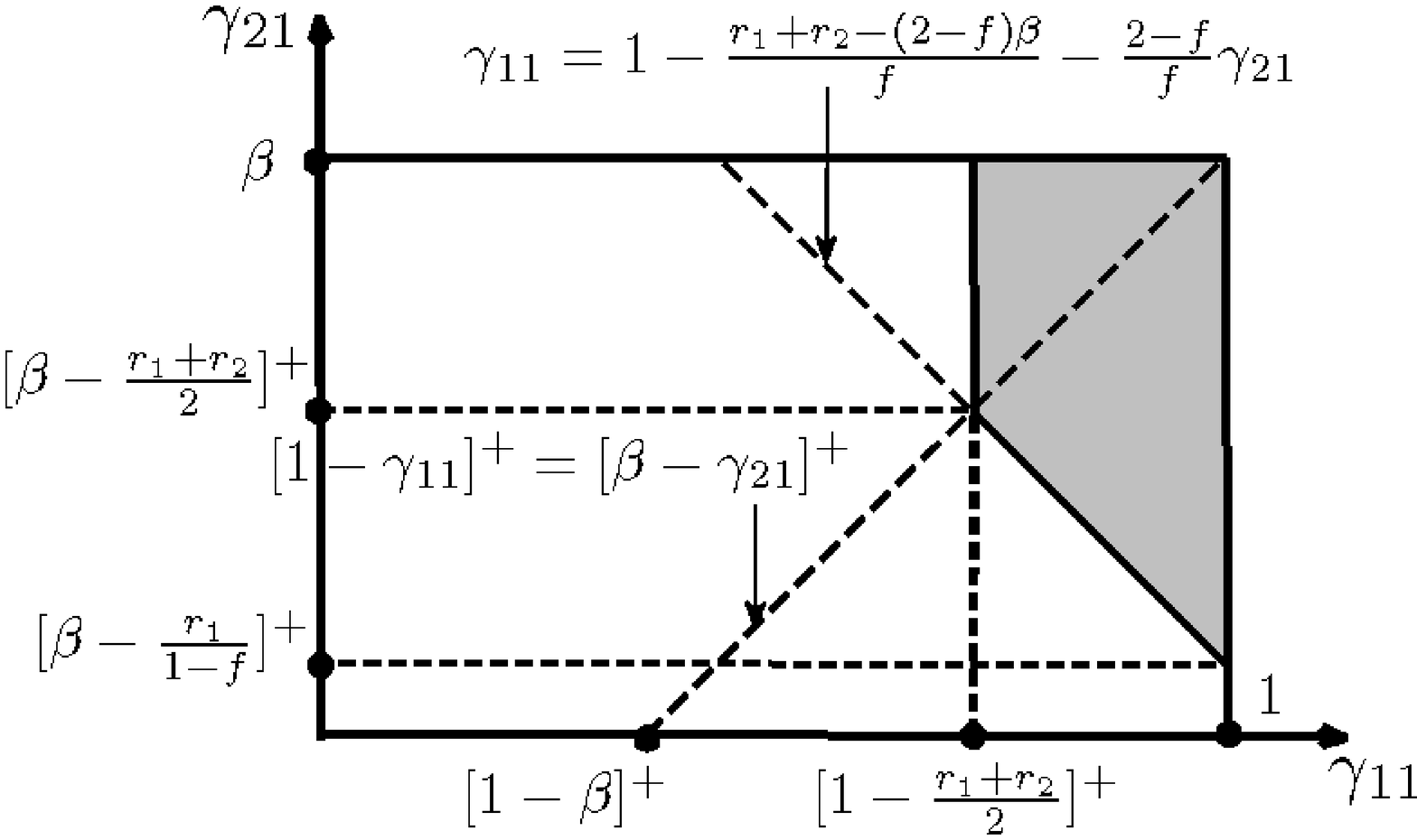}
 }
\caption{Constraint regions of $d_{12,\rm{CMO}}^c(2)$}
\label{fig:proof5-b}
\end{figure}

For the TIAN scheme, on the other hand, the outage event at RX1 at the end of the second transmission round can be expressed as follows.
\begin{equation}
\overline{\cal{A}}_2=\left\{{\cal{F}}_{T'},{\cal{O}}_3\right\},
\end{equation}
where,
\begin{equation}
\begin{split}
{\cal{O}}_3=\Bigg\{&h_{11},h_{21}:\\
&\;\log\left(1+\frac{|h_{11}|^2\rho}{1+|h_{21}|^2\rho^{\beta}}\right)+\frac{T'}{T}\log\left(1+|h_{11}|^2\rho\right)\\
&+\frac{T-T'}{T}\log\left(1+|h_{11}|^2\rho+|h_{21}|^2\rho^\beta\right)<R_1\Bigg\},
\end{split}
\label{eq:COP-STAT-OUT-TIAN}
\end{equation}
which can be reduced, in the high-$\rho$ limit, to
\begin{equation}
\begin{split}
{\cal{O}}_3=\Bigg\{&\gamma_{11},\gamma_{21},f:\\
&\;\left[1-\gamma_{11}-\left[\beta-\gamma_{21}\right]^+\right]^++f\left[1-\gamma_{11}\right]^+\\
&+(1-f)\max\left\{[1-\gamma_{11}]^+,[\beta-\gamma_{21}]^+\right\}<r_1\Bigg\},
\end{split}
\label{eq:COP-STAT-OUT-TIAN-HSNR}
\end{equation}

Similarly, we have
\begin{equation}
\begin{split}
{\rm{Pr}}\left(\overline{\cal{A}}_2\right)&=\underset{f\in[r_1,1]}\max{\rm{Pr}}({\cal{O}}_3)\\
&\doteq\rho^{-d_{1,\rm{TIAN}}^c(2)}.
\end{split}
\end{equation}

The shaded regions in Fig. \ref{fig:proof5-c} show the constraint regions of $\gamma_{11}$ and $\gamma_{21}$ for the cases $r_1\geq\beta$, $(1-f)\beta\leq r_1<\beta$, and $r_1<(1-f)\beta$. Thus, we have
\begin{equation}
\begin{split}
&d_{1,\rm{TIAN}}^c(2)=\underset{f\in[r_1,1]}\min\\
&\begin{cases}
\left[1-\frac{r_1+\beta}{2}\right]^++1-\frac{r_1}{f},\qquad\text{if}\;\; r_1\geq\beta\\
\min\left\{[1-r_1]^++[\beta-r_1]^+,[1-\frac{r_1-(1-f)\beta}{f}]^+\right\}+1-\frac{r_1}{f},\\
\qquad\qquad\qquad\qquad\qquad\text{if}\;\;(1-f)\beta\leq r_1<\beta\\
\min\left\{[1-r_1]^++[\beta-r_1]^+,\left[1+\beta-\frac{r_1}{1-f}\right]^+\right\}+1-\frac{r_1}{f},\\
\qquad\qquad\qquad\qquad\qquad\text{if}\;\;r_1<(1-f)\beta.
\end{cases}
\end{split}
\label{eq:d1TIAN_f}
\end{equation}
Through similar optimization over $f$ as in the CMO scheme, it can be shown that RX1 diversity of the TIAN scheme under the cooperative ARQ setting with a maximum of two transmission rounds is given by
\begin{equation}
\begin{split}
&d_{1,\rm{TIAN}}^c(2)=\\
&\begin{cases}
\left[1-\frac{r_1+\beta}{2}\right]^+,\qquad\text{if}\;\;r_1\geq\beta\\
2\left[1-r_1\right]^+,\qquad\text{if}\;\;r_1\leq\frac{\beta}{2},\;\beta\geq1\\
\left[1-r_1\right]^++\left[\beta-r_1\right]^+,\qquad\text{if}\;\;r_1\leq\frac{\beta}{2},\;\beta<1\\
\frac{\left(1-r_1\right)\beta}{r_1},\qquad\text{if}\;\;r_1>\frac{1}{2},\;\frac{\beta}{2}<r_1<\beta\\
\left[1-r_1\right]^++\left[\beta-r_1\right]^+,\qquad\text{if}\;\;r_1\leq\frac{1}{2},\;\frac{\beta}{2}<r_1<\beta.
\end{cases}
\end{split}
\end{equation}

\begin{figure}
\centering
\subfigure[$r_1\geq\beta$]{
   \includegraphics[width=76mm, height = 50mm] {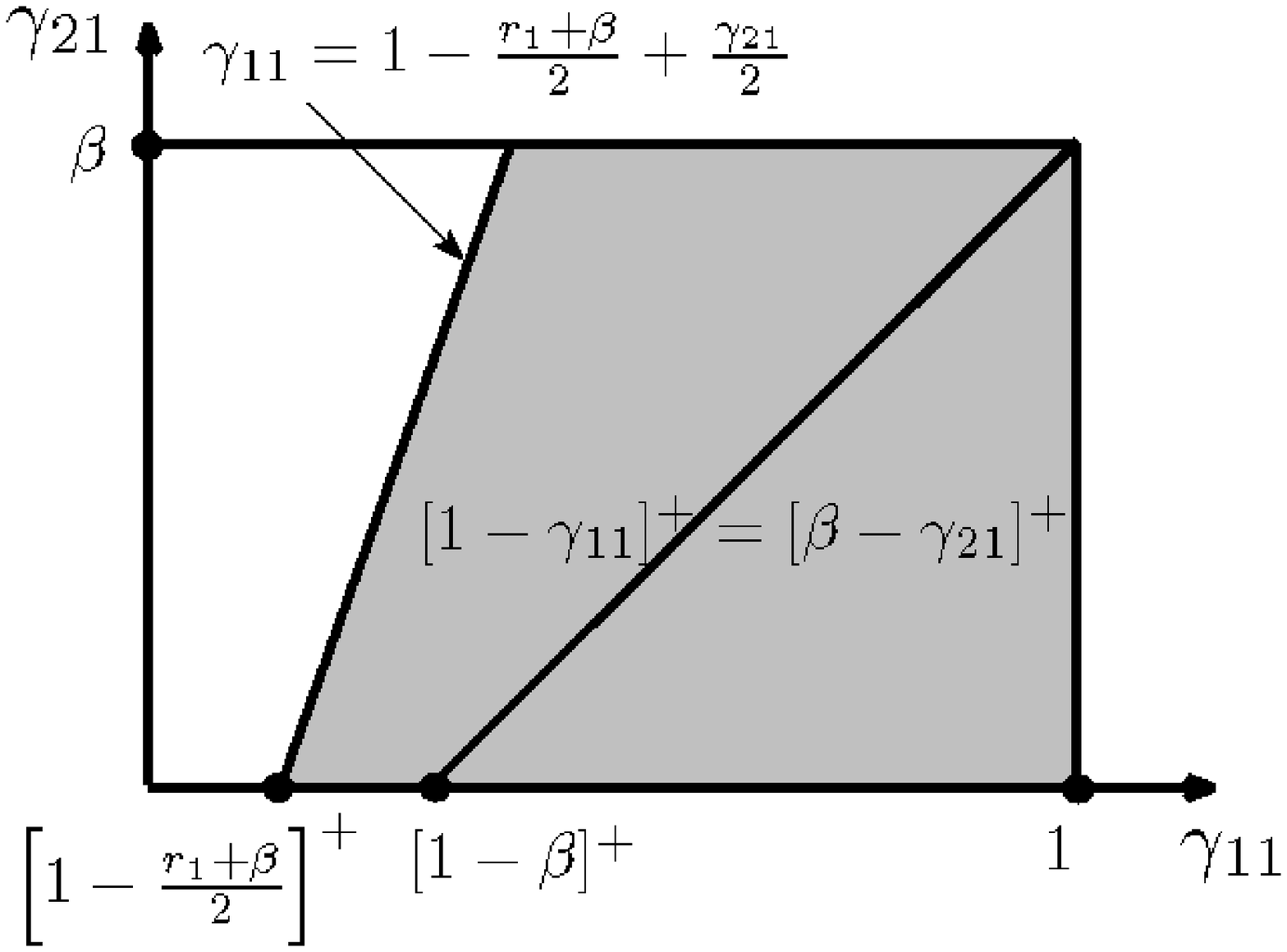}
 }
\subfigure[$(1-f)\beta\leq r_1<\beta$]{
   \includegraphics[width=76mm, height = 50mm] {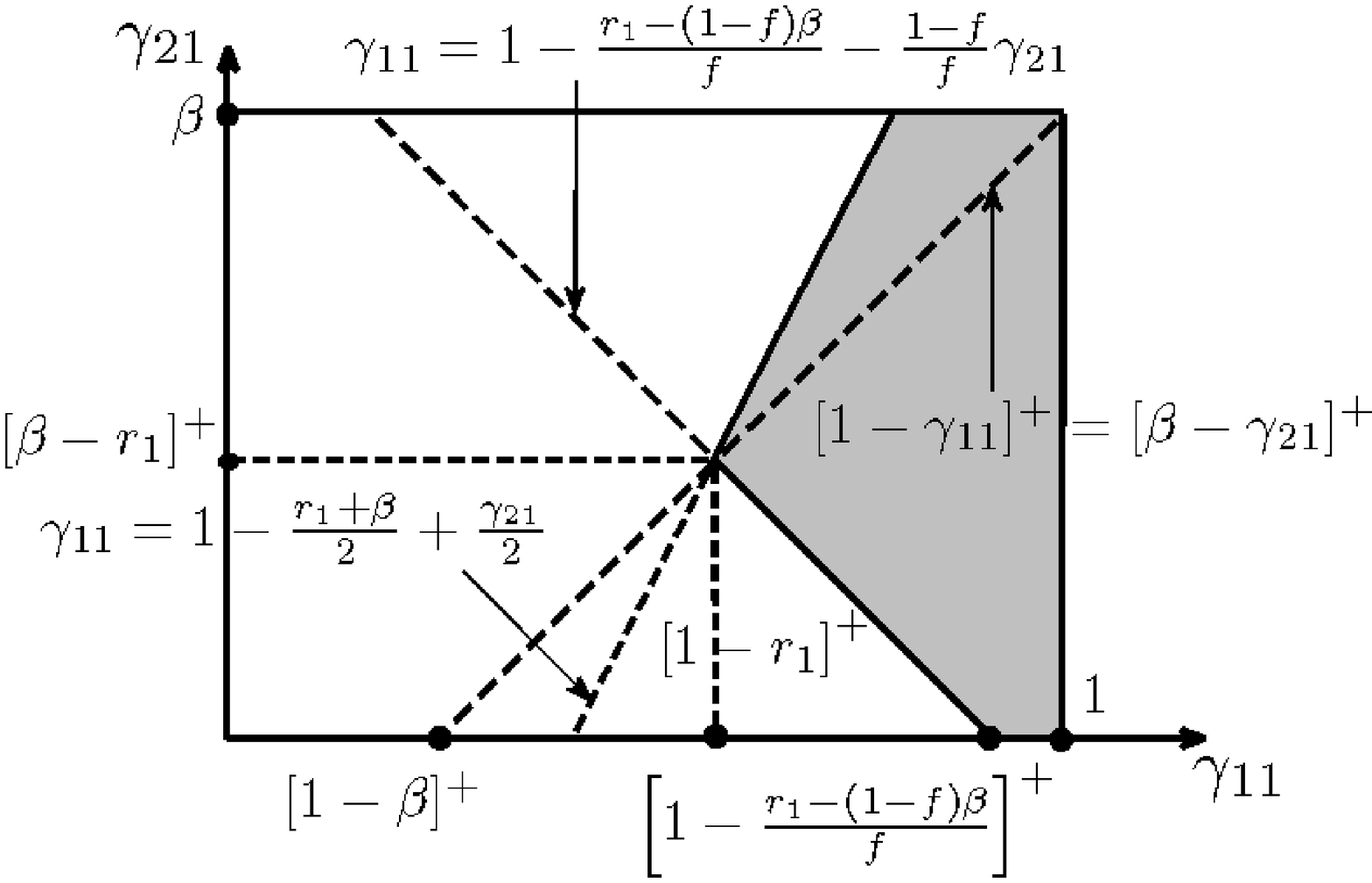}
 }
\subfigure[$r_1<(1-f)\beta$]{
   \includegraphics[width=76mm, height = 50mm] {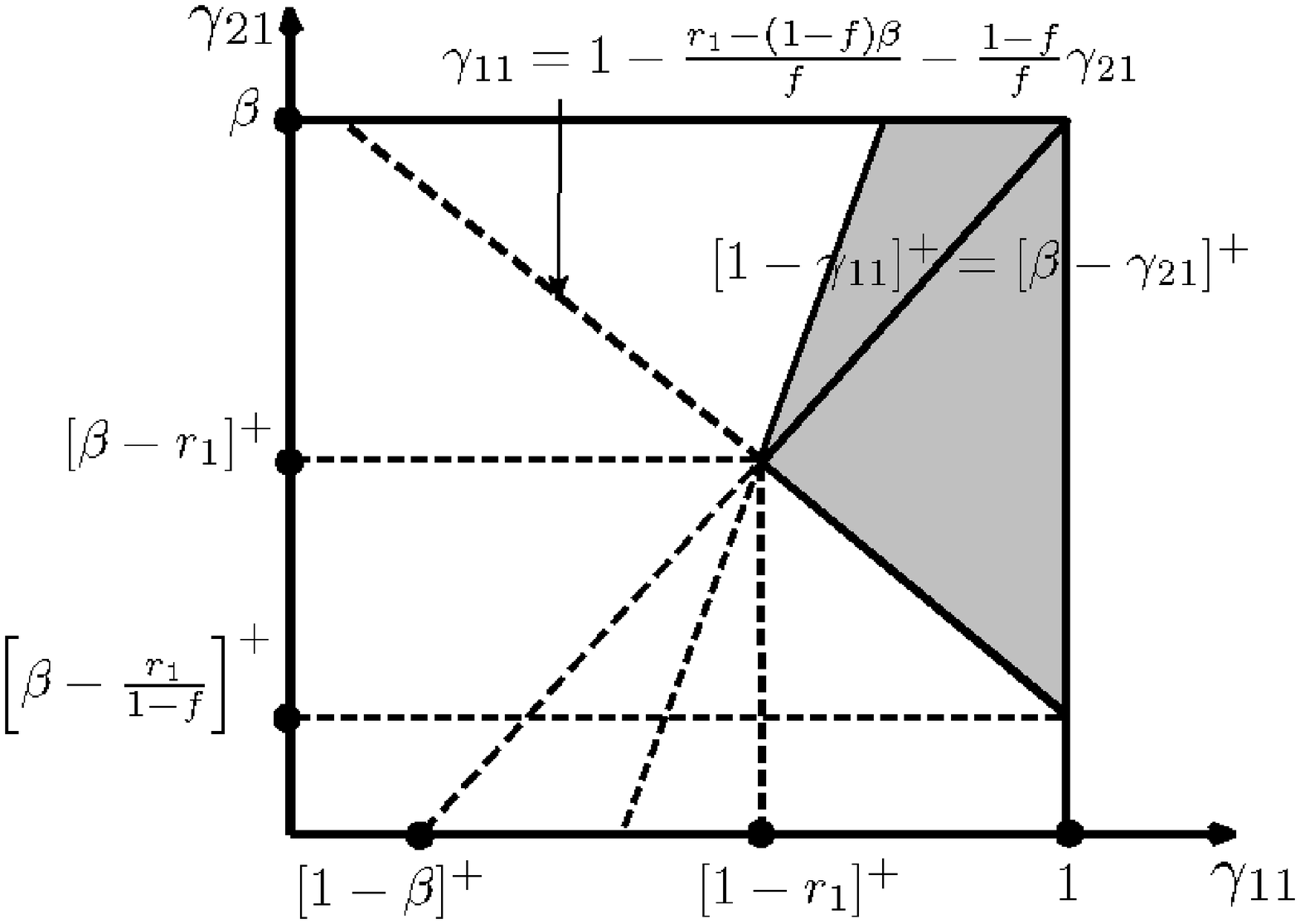}
 }
\caption{Constraint regions of $d_{1,\rm{TIAN}}^c(2)$}
\label{fig:proof5-c}
\end{figure}

Now, for both the CMO and TIAN schemes, the error event at RX2 $\left\{{\cal{E}}_2\right\}$ can be decomposed to the following events.
\begin{enumerate}
\item $\left\{{\cal{E}}_2,{\cal{A}}_{1}\right\}$ denotes the error event at RX2 when RX1 receives an ACK at the end of the first transmission round. This event is the union of three events; an undetected decoding error event at the end of the first transmission round $\left\{{\cal{E}}_2,{\cal{A}}_{1},{\cal{B}}_{1}\right\}$, a decoding failure event at the end of the second transmission round $\left\{{\cal{E}}_2,{\cal{A}}_{1},\overline{\cal{B}}_2\right\}$, and an undetected decoding error event at the end of the second transmission round $\left\{{\cal{E}}_2,{\cal{A}}_{1},{\cal{B}}_{2}\right\}$.
\item $\left\{{\cal{E}}_2,\overline{\cal{A}}_1\right\}$ denotes the error event at RX2 when RX1 receives a NACK at the end of the first transmission round. This event can be expressed as the union of two events; an undetected decoding error event $\left\{{\cal{E}}_2,\overline{\cal{A}}_1,{\cal{B}}_1\right\}$ and a decoding failure event  $\left\{{\cal{E}}_2,\overline{\cal{A}}_1,\overline{\cal{B}}_1\right\}$ at the end of the first transmission round.
\end{enumerate}

It can be shown that the dominating error events are the events $\left\{{\cal{E}}_2,{\cal{A}}_{1},\overline{\cal{B}}_2\right\}$ and $\left\{{\cal{E}}_2,\overline{\cal{A}}_1,\overline{\cal{B}}_1\right\}$. Thus,
\begin{equation}
\begin{split}
{\rm{Pr}}\left({\cal{E}}_2\right)&\doteq {\rm{Pr}}\left(\overline{\cal{A}}_1,\overline{\cal{B}}_1\right)+{\rm{Pr}}\left({\cal{A}}_1,\overline{\cal{B}}_2\right)\\
&\doteq {\rm{Pr}}\left(\overline{\cal{A}}_1\right) {\rm{Pr}}\left(\overline{\cal{B}}_1\right)+{\rm{Pr}}\left(\overline{\cal{B}}_2\right),
\end{split}
\label{eq:E_2}
\end{equation}
\noindent as the events $\overline{\cal{A}}_1$ and $\overline{\cal{B}}_1$ are independent, as well as the events ${\cal{A}}_1$ and $\overline{\cal{B}}_2$. Also, ${\rm{Pr}}({\cal{A}}_1)\doteq 1$.

For the CMO special case, it was shown in \cite{me} that the outage probabilities at RX1 and RX2 are given by
\begin{equation}
{\rm{Pr}}\left(\overline{\cal{A}}_1\right)\doteq\rho^{-\min\left\{\left[1-r_1\right]^+,\left[1-r_1-r_2\right]^++\left[\beta-r_1-r_2\right]^+\right\}}
\label{eq:A_1Bar}
\end{equation}
\begin{equation}
{\rm{Pr}}\left(\overline{\cal{B}}_1\right)\doteq\rho^{-[1-r_2]^+}.
\label{eq:B_1Bar}
\end{equation}
Also, using equation (\ref{eq:outage2CMO-HSNR}) with $L=2$ we have
\begin{equation}
{\rm{Pr}}\left(\overline{\cal{B}}_2\right)\doteq\rho^{-\left[1-\frac{r_2}{2}\right]^+}.
\label{eq:B_2Bar}
\end{equation}
Using equations (\ref{eq:A_1Bar}), (\ref{eq:B_1Bar}), and (\ref{eq:B_2Bar}) in (\ref{eq:E_2}), we have
\begin{equation}
\begin{split}
&{\rm{Pr}}\left({\cal{E}}_2\right)\doteq\\
&\rho^{-\min\left\{\min\left\{\left[1-r_1\right]^+,\left[1-r_1-r_2\right]^++\left[\beta-r_1-r_2\right]^+\right\}+[1-r_2]^+,\left[1-\frac{r_2}{2}\right]^+\right\}}\\
&\qquad\;\;\;\doteq\rho^{-d_{2,\rm{CMO}}^c(2)}.
\end{split}
\end{equation}
Thus,
\begin{equation}
\begin{split}
&d_{2,\rm{CMO}}^c(2)=\min\\
&\Bigg\{\min\left\{\left[1-r_1\right]^+,\left[1-r_1-r_2\right]^++\left[\beta-r_1-r_2\right]^+\right\}\\
&\qquad\qquad\qquad+[1-r_2]^+,\left[1-\frac{r_2}{2}\right]^+\Bigg\}.
\end{split}
\end{equation}

Using similar arguments, it is an easy matter to show that RX2 diversity of the TIAN special case under the cooperative ARQ setting with $L=2$ is given by
\begin{equation}
d_{2,\rm{TIAN}}^c(2)=\min\left\{[1-r_1-\beta]^++[1-r_2]^+,\left[1-\frac{r_2}{2}\right]^+\right\}.
\end{equation}
\end{proof}

\subsection{Cooperative ARQ with Dynamic Decoding}
We consider here a dynamic decoder as follows. Each time {\it{both}} TX1 and TX2 begin to transmit new messages, RX1 decides to use either the CMO or the TIAN decoding according to the channel conditions revealed to it, $h_{11}$ and $h_{21}$. The decoding scheme is no longer known a priori but is dynamically decided each time users transmit new messages. It is worthwhile noticing that the second transmitter has no CSI to dynamically change its splitting parameters according to the channel conditions. The first receiver RX1 has to determine either to decode the whole information (interference) sent by TX2 or to treat it as additive noise. We now state the achievable tradeoff of this approach with maximum of two transmission rounds in the following theorem.

\begin{theorem}
The achievable DMT of the cooperative ARQ with dynamic decoding scheme for $L=2$ is characterized as follows.
\begin{equation}
\begin{split}
&\qquad\qquad d_{1,\rm{DD}}^c(2)=\min\left\{d_{11,\rm{DD}}^c(2),d_{12,\rm{DD}}^c(2)\right\},\\
&\text{where,}\\
&d_{11,\rm{DD}}^c(2)=d_{11,\rm{CMO}}^c(2)\\
&d_{12,\rm{DD}}^c(2)=\\
&\begin{cases}
d_{1,\rm{TIAN}}^c(2),\qquad r_2\geq\beta\\
d_{12,\rm{CMO}}^c(2),\qquad r_2<\beta,\;r_1\geq r_2\\
\left[\beta-\frac{(2r_1-1)r_2}{r_1}\right]^+,\qquad r_2<\beta,\;\frac{1}{2}\leq r_1<r_2\\
[1-r_1]^++[\beta-r_1]^+,\;\;\;r_2<\beta,\;r_1<\min\left\{\frac{1}{2},r_2\right\}.
\end{cases}\\
&\text{And,}\\
&\qquad\qquad d_{2,\rm{DD}}^c(2)=\min\left\{d_{21,\rm{DD}}^c(2),d_{22,\rm{DD}}^c(2)\right\},\\
&\text{where,}\\
&d_{21,\rm{DD}}^c(2)=\left[1-r_2\right]^++\max\left\{d_{1,\rm{CMO}}(1),d_{1,\rm{TIAN}}(1)\right\}\\
&d_{22,\rm{DD}}^c(2)=\left[1-\frac{r_2}{2}\right]^+,
\end{split}
\label{eq:DDdiversities}
\end{equation}
where, $d_{1,\rm{CMO}}(1)$, $d_{1,\rm{TIAN}}(1)$, $d_{12,\rm{CMO}}^c(2)$, and $d_{1,\rm{TIAN}}^c(2)$ are as given in (\ref{eq:CMO1DMT}), (\ref{eq:TIAN1DMT}), (\ref{eq:COPSTATCMO}), and (\ref{eq:COPSTATTIAN}), respectively.
\end{theorem}

\begin{proof}
Based on the dynamic decoding scheme of the cooperative ARQ protocol, outage at RX1 at the end of the second transmission round can be described as follows.
\begin{equation}
\overline{\cal{A}}_2=\left\{{\cal{F}}_{T'},\left\{{\cal{O}}_1\cup{\cal{O}}_2\right\},{\cal{O}}_3\right\},
\end{equation}
where, ${\cal{F}}_{T'}$, ${\cal{O}}_1$, ${\cal{O}}_2$, and ${\cal{O}}_3$ are as defined in (\ref{eq:COP-STAT-F-CMO-HSNR}), (\ref{eq:COP-STAT-OUT1-CMO-HSNR}), (\ref{eq:COP-STAT-OUT2-CMO-HSNR}), and (\ref{eq:COP-STAT-OUT-TIAN-HSNR}), respectively.

Thus,
\begin{equation}
\begin{split}
&\overline{\cal{A}}_2=\left\{{\cal{O}}_{1,\rm{DD}}\cup{\cal{O}}_{2,\rm{DD}}\right\}\\
&\text{where,}\\
&{\cal{O}}_{11,\rm{DD}}=\left\{{\cal{F}}_{T'},{\cal{O}}_1,{\cal{O}}_3\right\}\\
&{\cal{O}}_{12,\rm{DD}}=\left\{{\cal{F}}_{T'},{\cal{O}}_2,{\cal{O}}_3\right\}.
\end{split}
\end{equation}

And hence, the following relations hold.
\begin{equation}
\begin{split}
{\rm{Pr}}\left(\overline{\cal{A}}_2\right)&={\rm{Pr}}({\cal{O}}_{11,\rm{DD}}\cup{\cal{O}}_{12,\rm{DD}})\\
&\doteq\rho^{-d_{1,\rm{DD}}^c(2)},
\end{split}
\end{equation}
where,
\begin{equation}
\begin{split}
&{\rm{Pr}}({\cal{O}}_{11,\rm{DD}})\doteq\rho^{-d_{11,\rm{DD}}^c(2)}\\
&{\rm{Pr}}({\cal{O}}_{12,\rm{DD}})\doteq\rho^{-d_{12,\rm{DD}}^c(2)}.
\end{split}
\end{equation}
thus,
\begin{equation}
d_{1,\rm{DD}}^c(2)=\min\left\{d_{11,\rm{DD}}^c(2),d_{12,\rm{DD}}^c(2)\right\}.
\end{equation}

We can notice that ${\cal{O}}_1\subset{\cal{O}}_3$, thus; ${\cal{O}}_{11,\rm{DD}}=\left\{{\cal{F}}_{T'},{\cal{O}}_1\right\}$. Which yields that $d_{11,\rm{DD}}^c(2)=d_{11,\rm{CMO}}^c(2)$.

Now, in order to find an expression for $d_{12,\rm{DD}}^c(2)$, we have to solve the following optimization problem.
\begin{equation}
d_{12,\rm{DD}}^c(2)=\underset{\gamma_{11},\gamma_{21},u\in{\cal{O}}_{12,\rm{DD}}}\min\left\{\gamma_{11}+\gamma_{21}+u\right\},
\end{equation}
where $u$ is defined as given in (\ref{eq:u}).

The shaded regions in Fig. \ref{fig:proofdd} show the constraint regions for the different MGRs $(r_1,r_2)$. Thus,
\begin{equation}
\begin{split}
&d_{12,\rm{DD}}^c(2)=\underset{f\in[r_1,1]}\min\\
&\begin{cases}
\left[1-\frac{r_1+\beta}{2}\right]^++1-\frac{r_1}{f},\qquad\text{if}\;\; r_1\geq\beta,\;r_2\geq\beta\\
\min\left\{[1-r_1]^++[\beta-r_1]^+,\left[1-\frac{r_1-(1-f)\beta}{f}\right]^+\right\}\\
\;+1-\frac{r_1}{f},\;\;\text{if}\;(1-f)\beta\leq r_1<\beta,\;r_2\geq\beta,\;r_1+r_2\geq2\beta\\
\min\left\{[1-r_1]^++[\beta-r_1]^+,1+\left[1+\beta-\frac{r_1}{1-f}\right]^+\right\}\\
\;+1-\frac{r_1}{f},\qquad\text{if}\;\;r_1<(1-f)\beta,\;r_2\geq\beta,\;r_1+r_2\geq2\beta\\
2-\frac{r_1+r_2}{2}-\frac{r_1}{f},\qquad\text{if}\;\;r_1\geq\beta,\;r_2<\beta,\;r_1+r_2\geq2\beta\\
\left[1-\frac{r_1+r_2}{2}\right]^++\left[\beta-\frac{r_1+r_2}{2}\right]^++1-\frac{r_1}{f},\\
\qquad\qquad\text{if}\;\;r_2<\beta,\;r_1\geq r_2,\;r_1+r_2<2\beta\\
\min\left\{[1-r_1]^++[\beta-r_1]^+,\left[1+\beta-\frac{r_1}{f}+\frac{1-2f}{f}r_2\right]^+\right\},\\
\qquad\qquad\text{if}\;r_1<\beta,\;r_2<\beta,\;r_1<r_2.
\end{cases}
\end{split}
\end{equation}
We notice that for $r_2\geq\beta$, $d_{12,\rm{DD}}^c(2)$ is similar to $d_{1,\rm{TIAN}}^c(2)$ given in equation (\ref{eq:d1TIAN_f}); thus, a similar minimization over $f$ can be performed. Also, for $r_2<\beta,\;r_1\geq r_2$, $d_{12,\rm{DD}}^c(2)$ is the same as $d_{12,\rm{CMO}}^c(2)$ given in equation (\ref{eq:d12CMO_f}). Finally, we perform minimization over $f$ for the function $\min\left\{[1-r_1]^++[\beta-r_1]^+,\left[1+\beta-\frac{r_1}{f}+\frac{1-2f}{f}r_2\right]^+\right\}$ using similar steps. Based on these arguments, we have
\begin{equation}
\begin{split}
&d_{12,\rm{DD}}^c(2)=\\
&\begin{cases}
d_{1,\rm{TIAN}}^c(2),\qquad r_2\geq\beta\\
d_{12,\rm{CMO}}^c(2),\qquad r_2<\beta,\;r_1\geq r_2\\
\left[\beta-\frac{(2r_1-1)r_2}{r_1}\right]^+,\qquad r_2<\beta,\;\frac{1}{2}\leq r_1<r_2\\
[1-r_1]^++[\beta-r_1]^+,\;\;\;r_2<\beta,\;r_1<\min\left\{\frac{1}{2},r_2\right\}.
\end{cases}
\end{split}
\end{equation}

For the second user, similar arguments as those for the static decoding approach in the previous subsection hold but with a little difference. For the dynamic decoding approach, RX1 either use the CMO or the TIAN form of decoding. Therefore, the probability of the outage event at RX1 at the end of the first transmission round is the maximum of what we can get using the CMO scheme and what we can get using the TIAN scheme. So, ${\rm{Pr}}\left(\overline{\cal{A}}_1\right)$ becomes as follows.
\begin{equation}
{\rm{Pr}}\left(\overline{\cal{A}}_1\right)\doteq\rho^{-\max\left\{d_{1,\rm{CMO}}(1),d_{1,\rm{TIAN}}(1)\right\}},
\label{eq:A_1Bar_DD}
\end{equation}
where $d_{1,\rm{CMO}}(1)$ and $d_{1,\rm{TIAN}}(1)$ are as given in (\ref{eq:CMO1DMT}) and (\ref{eq:TIAN1DMT}), respectively. Similarly, using equations (\ref{eq:B_1Bar}), (\ref{eq:B_2Bar}), and (\ref{eq:A_1Bar_DD}) in equation (\ref{eq:E_2}), we get
\begin{equation}
\begin{split}
{\rm{Pr}}({\cal{E}}_2)&\doteq\rho^{-\min\left\{\max\left\{d_{1,\rm{CMO}}(1),d_{1,\rm{TIAN}}(1)\right\}+[1-r_2]^+,\left[1-\frac{r_2}{2}\right]^+\right\}}\\
&\doteq\rho^{-d_{2,\rm{DD}}^c(2)}.
\end{split}
\end{equation}
Thus,
\begin{equation}
\begin{split}
&d_{2,\rm{DD}}^c(2)=\min\\
&\left\{\max\left\{d_{1,\rm{CMO}}(1),d_{1,\rm{TIAN}}(1)\right\}+[1-r_2]^+,\left[1-\frac{r_2}{2}\right]^+\right\}.
\end{split}
\end{equation}

\begin{figure}
\centering
\subfigure[$r_1\geq\beta,\;r_2\geq\beta$]{
   \includegraphics[width=75mm, height = 35mm] {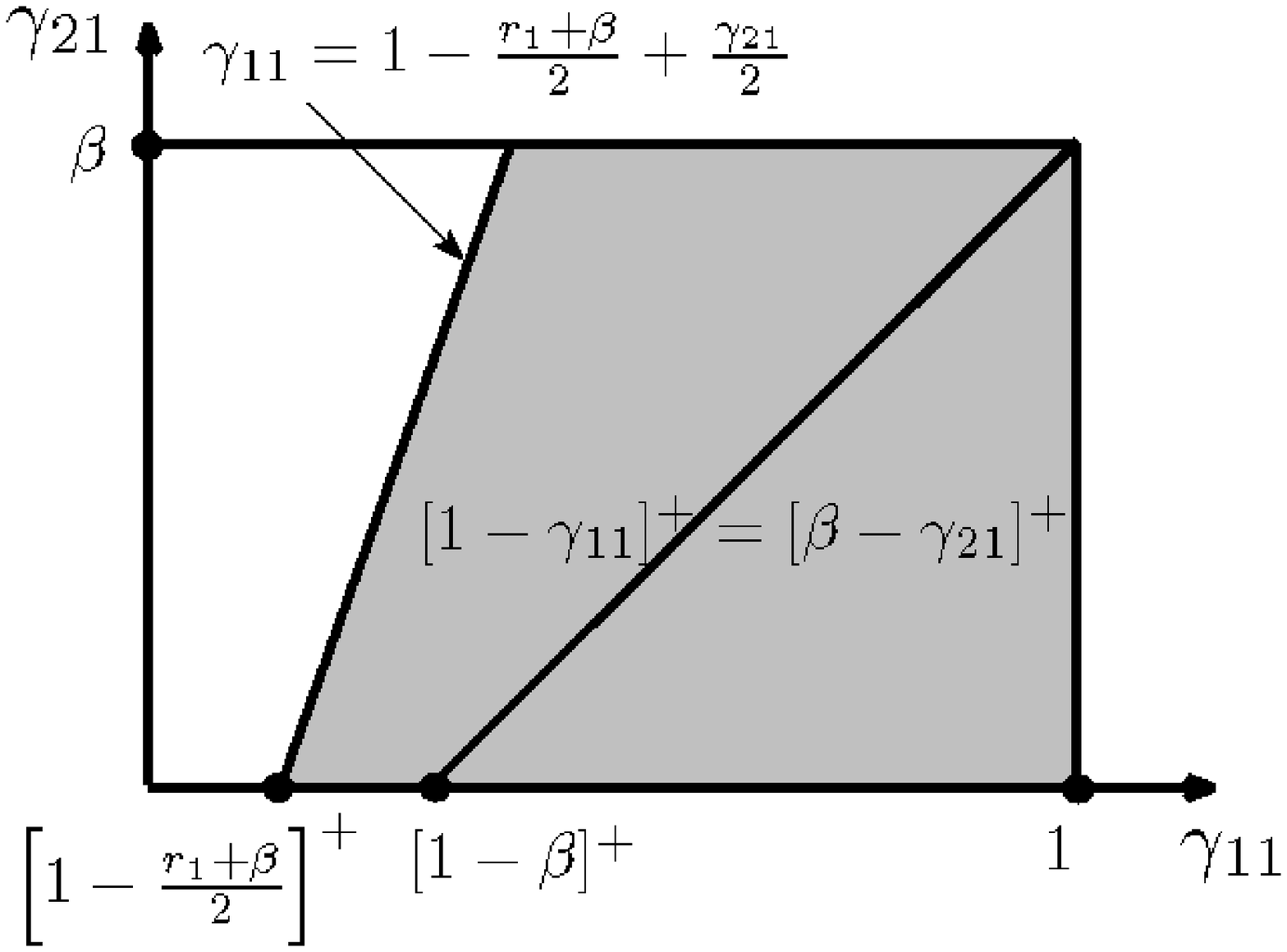}
 }
\subfigure[$(1-f)\beta\leq r_1<\beta,\;r_2\geq\beta,\;r_1+r_2\geq2\beta$]{
   \includegraphics[width=75mm, height = 35mm] {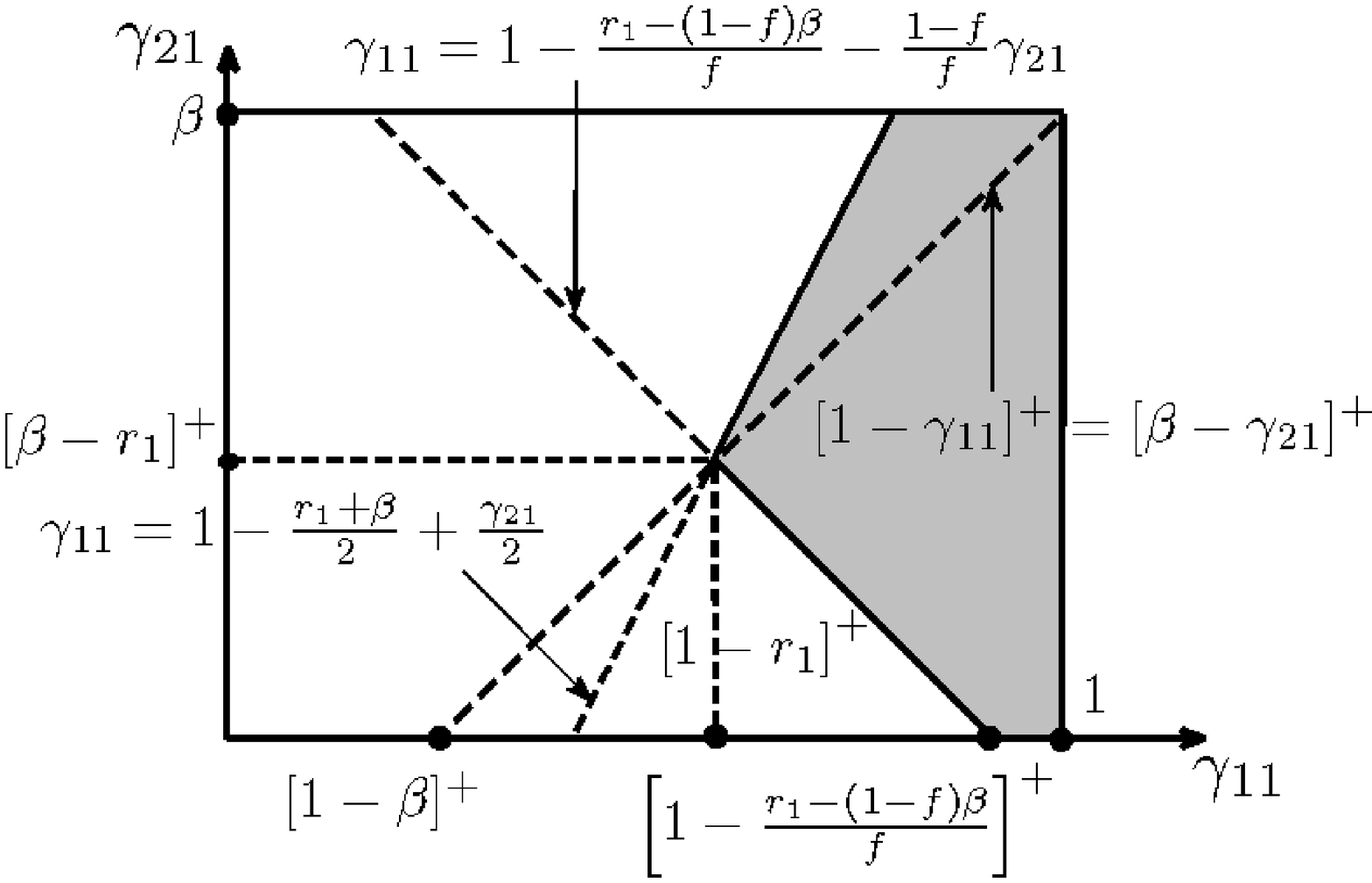}
 }
\subfigure[$r_1<(1-f)\beta,\;r_2\geq\beta,\;r_1+r_2\geq2\beta$]{
   \includegraphics[width=75mm, height = 35mm] {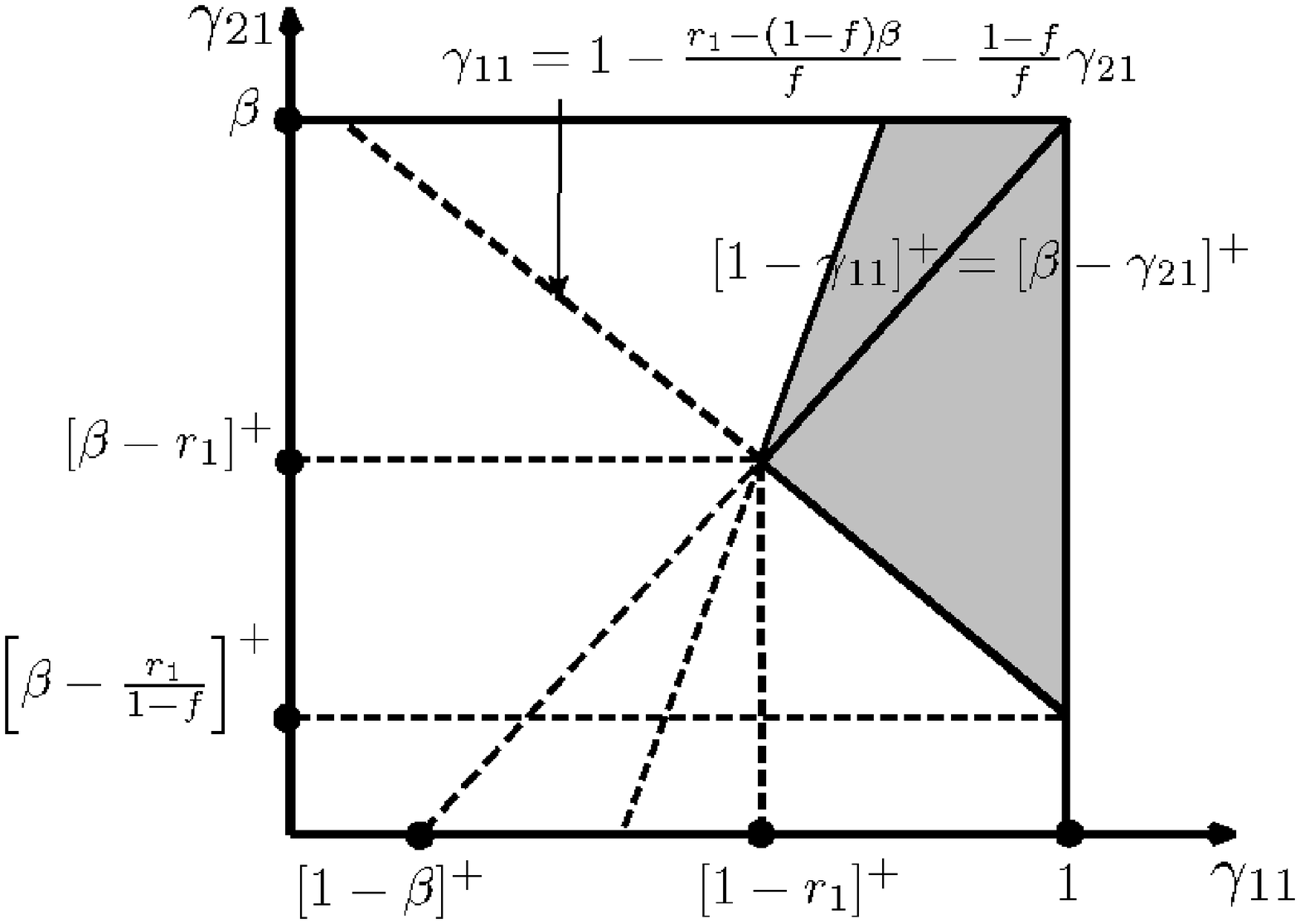}
 }
\subfigure[$r_1\geq\beta,\;r_2<\beta,\;r_1+r_2\geq2\beta$]{
   \includegraphics[width=75mm, height = 35mm] {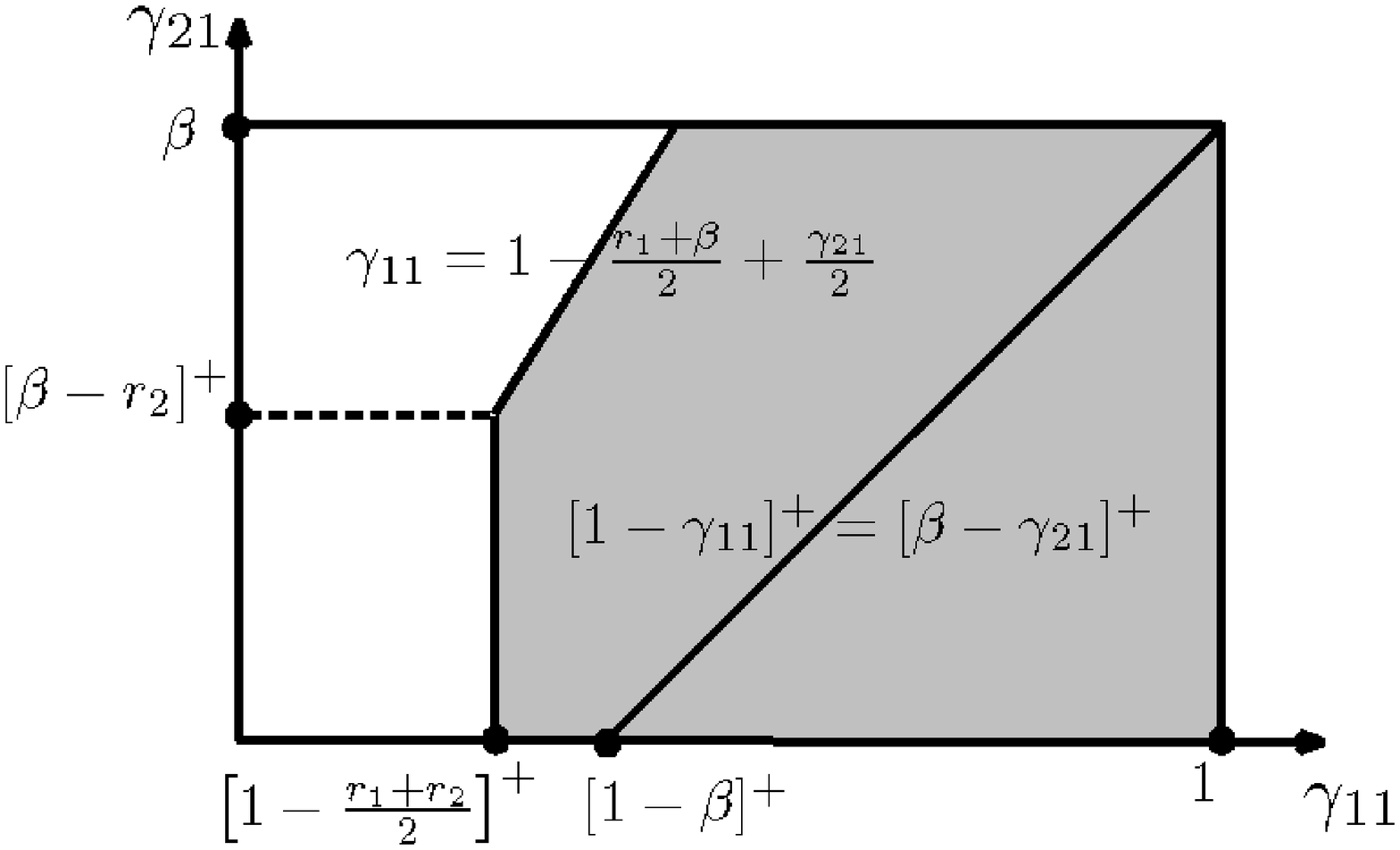}
 }
\subfigure[$r_2<\beta,\;r_1\geq r_2,\;r_1+r_2<2\beta$]{
   \includegraphics[width=75mm, height = 35mm] {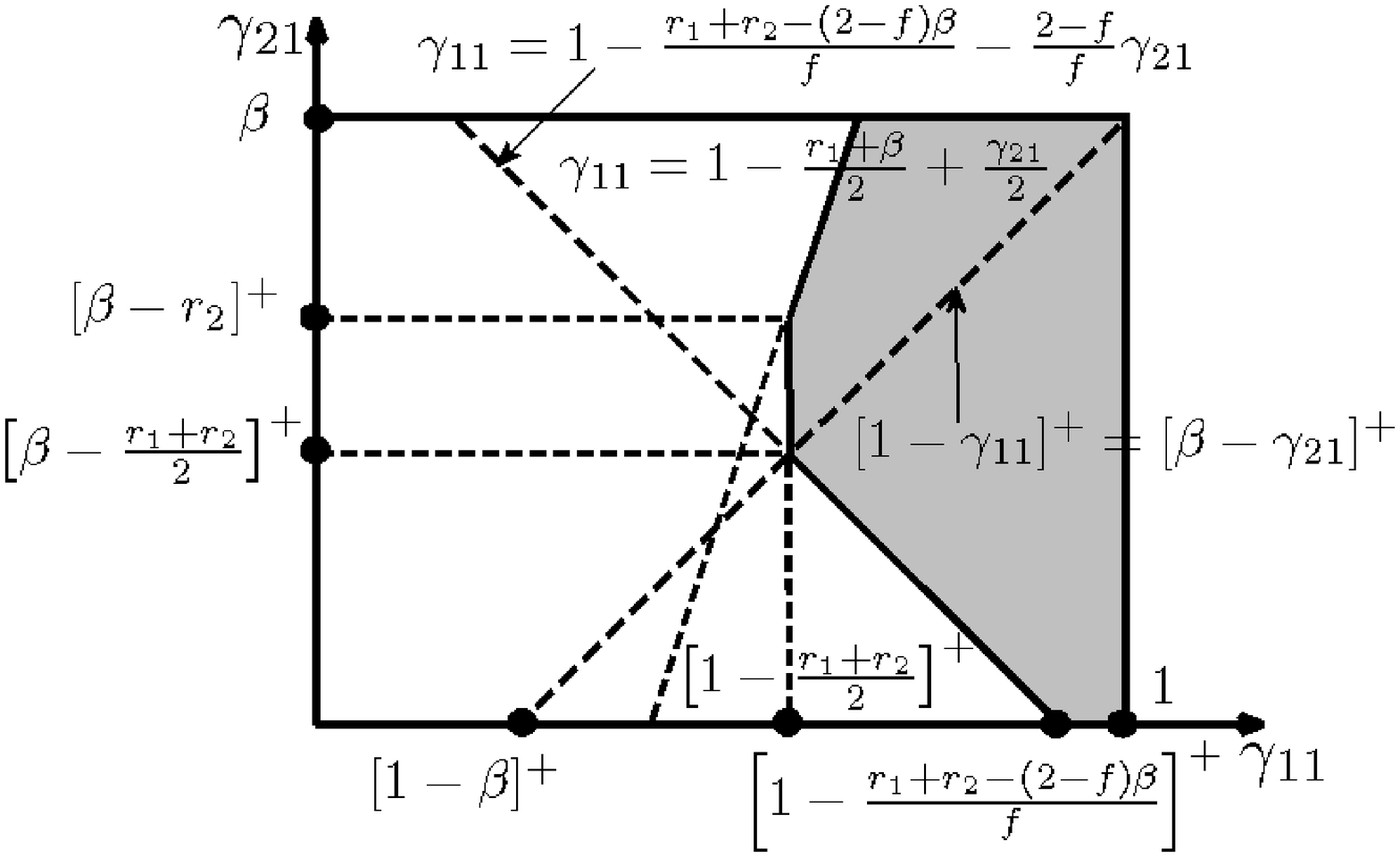}
 }
\subfigure[$r_1<{\beta},\;r_2<\beta,\;r_1<r_2$]{
   \includegraphics[width=75mm, height = 35mm] {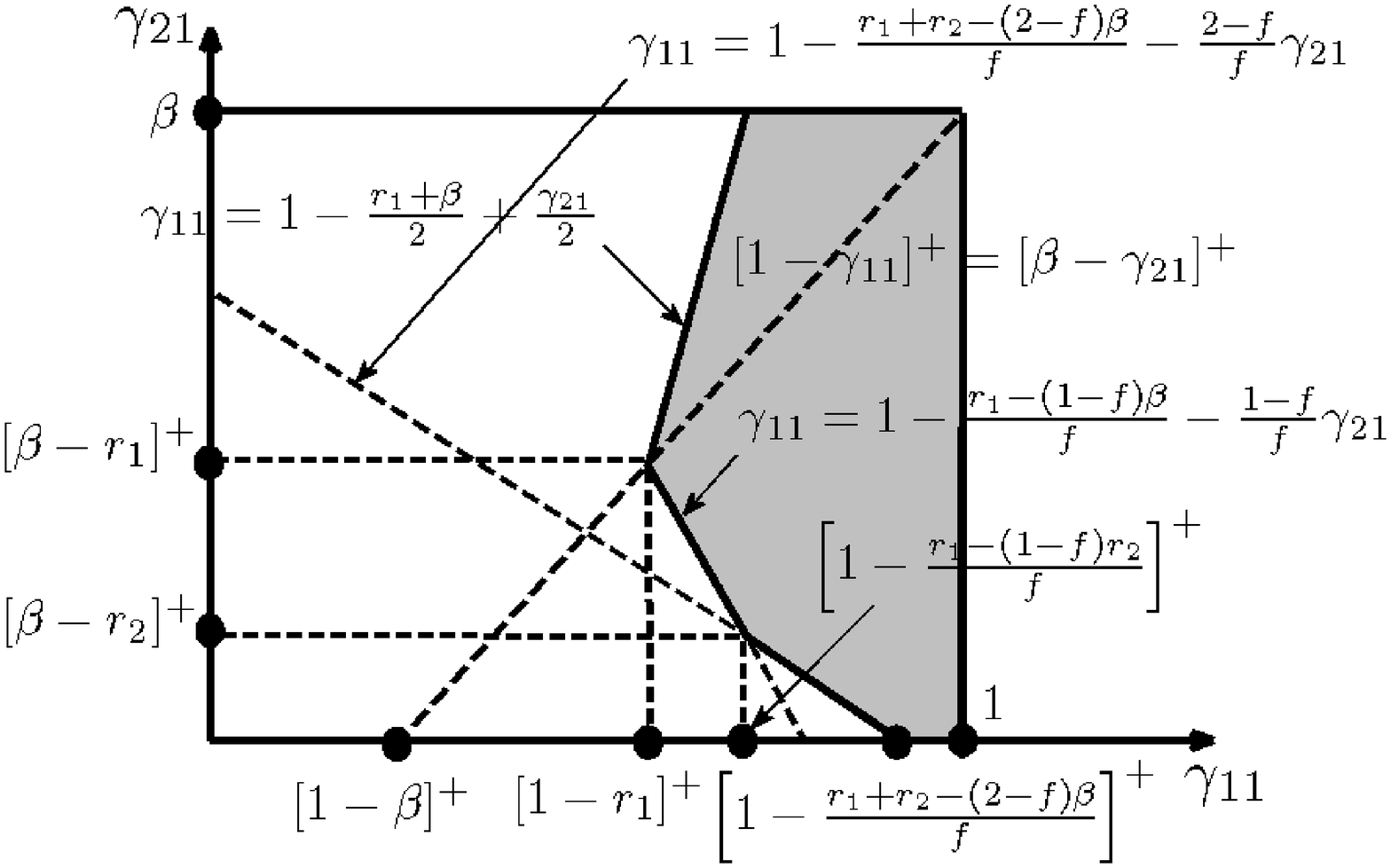}
 }
\caption{Constraint Regions for $d_{12,\rm{DD}}^c(2)$}
\label{fig:proofdd}
\end{figure}

\end{proof}

To summarize our work, we show the DMT of the first user under the use of all the previously mentioned ARQ schemes for $L=2$ in Fig. \ref{DD}. It is obvious in Fig. \ref{DD} that the performance of the cooperative ARQ with dynamic decoding is better than the overall achievable performance of its static counterpart for some values of the first user multiplexing gain, $r_1$.
\begin{figure}
	\centering
	\includegraphics[width=75mm, height = 50mm] {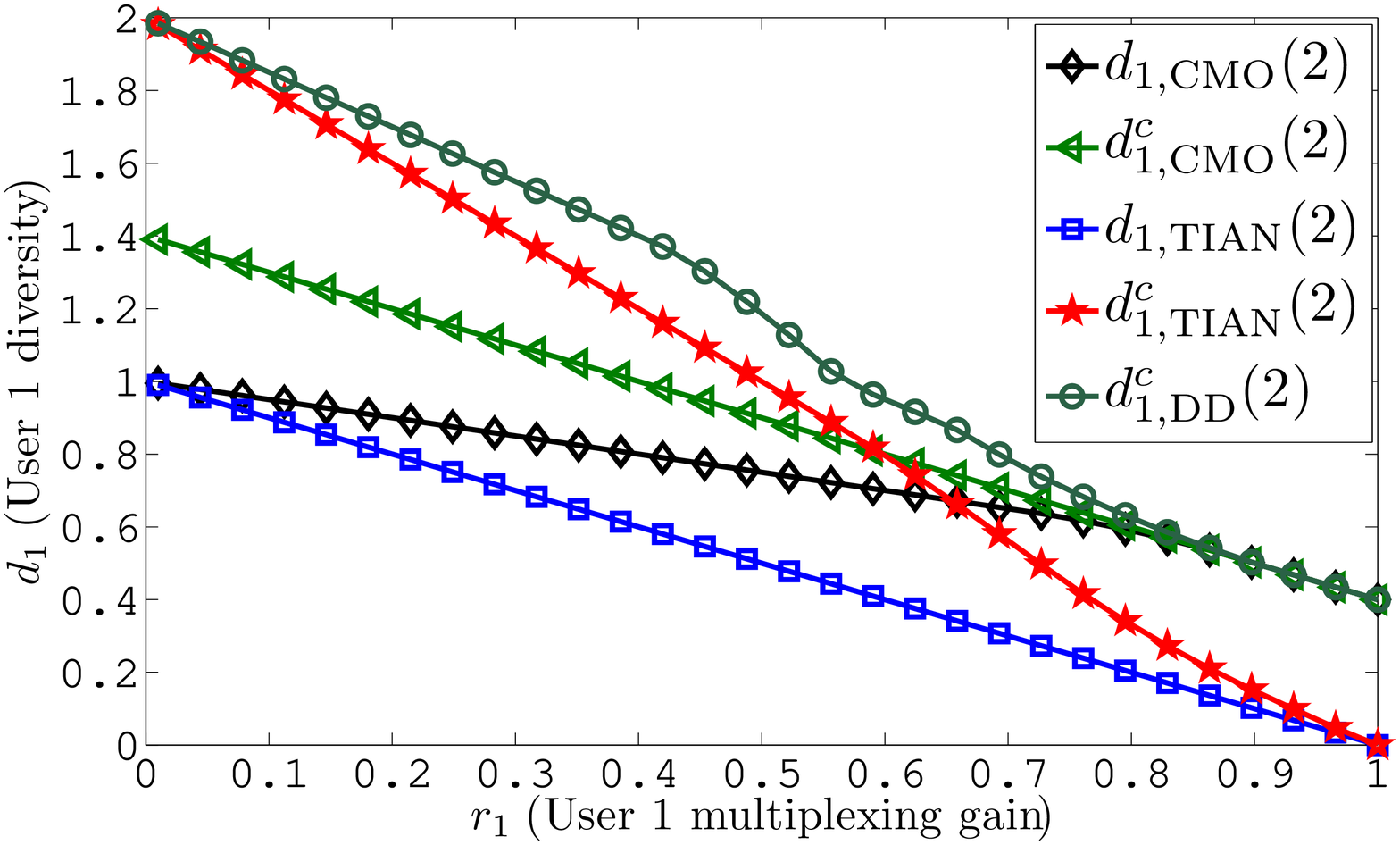}
	\caption{The DMT of user 1 using the different non-cooperative and cooperative ARQ schemes for $\beta=1.3,\;r_2=0.9$, and $L=2$.}
	\label{DD}
\vspace{-.2in}
\end{figure}

It is worthwhile noticing that RX2 diversity resulted from using cooperative dynamic decoding scheme is the maximum of RX2 diversity we get when using the cooperative CMO scheme and that we get using the cooperative TIAN scheme. This is obvious from equations (\ref{eq:COPSTATCMO}), (\ref{eq:COPSTATTIAN}), and (\ref{eq:DDdiversities}). Also, RX2 diversity under cooperation is always less than or equal to that of the non-cooperative ARQ scheme. Intuitively, the outage event of the second user under cooperation is more likely than that under no cooperation. Under cooperation, the second user lose the opportunity to retransmit its codeword in the event of a NACK reception by the both users since the second user will relay the first user's message starting from the second transmission round in this event. This can be shown analytically from equations (\ref{eq:CMO2DMT}), (\ref{eq:TIAN2DMT}), (\ref{eq:COPSTATCMO}), (\ref{eq:COPSTATTIAN}), and (\ref{eq:DDdiversities}). For $L=2$, RX2 diversity of the non-cooperative schemes is equal to $\left[1-\frac{r_2}{2}\right]^+$ while for the cooperative schemes, RX2 diversity equals to the minimum of $\left[1-\frac{r_2}{2}\right]^+$ and $d_{1}(1)+[1-r_2]^+$, where $d_{1}(1)$ is RX1 diversity of a single transmission round \cite{me}.

\section{Conclusion}\label{Con}
In this paper, we characterized the achievable diversity, multiplexing, and delay tradeoff of a two user single antenna  Rayleigh fading ARQ ZIC under the use of two different ARQ protocols. We used the well-known HK approach as well as two special cases of it, only a common or a private message is transmitted, to derive the achievability result. Our characterization of the three dimensional tradeoff comes in closed-form expressions of the individual diversities as a function of the maximum number of  transmission rounds (maximum delay), multiplexing gain pairs, interference level, and the rate and power splitting parameters.

\appendix
The outage region at RX1 given ${\cal{C}}_i$ at round $l$ for the non-cooperative ARQ ZIC system with maximum of $L$ transmission rounds is given in equation (\ref{eq:outage1}), where ${\cal{C}}_i$ is as defined in Section. \ref{coop}. The high-$\rho$ approximation of this outage region at round $l=L$ can be written as
\begin{equation}
\overline{\cal{A}}_L|{\cal{C}}_i=\left\{{\cal{O}}_{11,\rm{HK}}(L,i)\cup{\cal{O}}_{12,\rm{HK}}(L,i)\right\}
\end{equation}
where,
\begin{equation}
\begin{split}
&{\cal{O}}_{11,\rm{HK}}(L,i)=\\
&\Bigg\{\gamma_{11},\gamma_{21}:\;i\left[1-\gamma_{11}-\left[\beta-\gamma_{21}-b\right]^+\right]^+\\
&\qquad\qquad\qquad+(L-i)\left[1-\gamma_{11}\right]^+)<r_1\Bigg\}
\end{split}
\label{eq:outage11-HSNR}
\end{equation}
\begin{equation}
\begin{split}
&{\cal{O}}_{12,\rm{HK}}(L,i)=\\
&\Bigg\{\gamma_{11},\gamma_{21}:\;i\left[\max\left\{1-\gamma_{11},\beta-\gamma_{21}\right\}-\left[\beta-\gamma_{21}-b\right]^+\right]^+\\
&\qquad\qquad\;+(L-i)\max\left\{[1-\gamma_{21}]^+,[\beta-\gamma_{21}]^+\right\}<r_1+t_2\Bigg\}.
\end{split}
\label{eq:outage12-HSNR}
\end{equation}
The outage probability at RX1 given ${\cal{C}}_i$ at round $l=L$ can then be given by
\begin{equation}
\begin{split}
{\rm{Pr}}(\overline{\cal{A}}_L|{\cal{C}}_i)&={\rm{Pr}}\left\{{\cal{O}}_{11,\rm{HK}}(L,i)\cup{\cal{O}}_{12,\rm{HK}}(L,i)\right\}\\
&\doteq\rho^{-d_{\overline{\cal{A}}_L|{\cal{C}}_i}},
\end{split}
\label{eq:P(AL/Ci)1}
\end{equation}
where,
\begin{equation}
\begin{split}
&{\rm{Pr}}({\cal{O}}_{11,\rm{HK}}(L,i))\doteq\rho^{-d_{11,\rm{HK}}(L,i)}\\
&{\rm{Pr}}({\cal{O}}_{12,\rm{HK}}(L,i))\doteq\rho^{-d_{12,\rm{HK}}(L,i)}
\end{split}
\end{equation}
Therefore, we have
\begin{equation}
\begin{split}
{\rm{Pr}}(\overline{\cal{A}}_L|{\cal{C}}_i)&\doteq{\rm{Pr}}({\cal{O}}_{11,\rm{HK}}(L,i))+{\rm{Pr}}({\cal{O}}_{12,\rm{HK}}(L,i))\\
&\doteq\rho^{-\min\left\{d_{11,\rm{HK}}(L,i),d_{12,\rm{HK}}(L,i)\right\}},
\end{split}
\label{eq:P(AL/Ci)2}
\end{equation}
From (\ref{eq:P(AL/Ci)1}) and (\ref{eq:P(AL/Ci)2}), we have
\begin{equation}
d_{\overline{\cal{A}}_L|{\cal{C}}_i}=\min\left\{d_{11,\rm{HK}}(L,i),d_{12,\rm{HK}}(L,i)\right\}.
\end{equation}

Firstly, in order to find an expression for $d_{11,\rm{HK}}(L,i)$, we have to solve the following minimization problem.
\begin{equation}
d_{11,\rm{HK}}(L,i)=\underset{\gamma_{11},\gamma_{21}\in{\cal{O}}_{11,\rm{HK}}}\min\left\{\gamma_{11}+\gamma_{21}\right\}.
\end{equation}

The shaded regions in Fig. \ref{fig:proof4-a} show the constraint regions of $\gamma_{11}$ and $\gamma_{21}$ for $d_{11,\rm{HK}}(L,i)$ in the cases $r_1\geq(L-i)(\beta-b)$ and $r_1<(L-i)(\beta-b)$. Thus, we have
\begin{equation}
\begin{split}
d_{11,\rm{HK}}(L,i)=\begin{cases}
\left[1-\frac{r_1+i[\beta-b]^+}{L}\right]^+\text{if}\;\; r_1\geq(L-i)(\beta-b)\\
\left[1-\frac{r_1}{L-i}\right]^+,\qquad\text{if}\;\;r_1<(L-i)(\beta-b),
\end{cases}
\end{split}
\end{equation}
which yields
\begin{equation}
\begin{split}
&d_{11,\rm{HK}}(L,i)=\\
&\max\left\{\left[1-\frac{r_1+i[\beta-b]^+}{L}\right]^+,\left[1-\frac{r_1}{L-i}\right]^+\right\}.
\end{split}
\end{equation}

\begin{figure}
\centering
\subfigure[$r_1\geq(L-i)(\beta-b)$]{
   \includegraphics[width=75mm, height = 40mm] {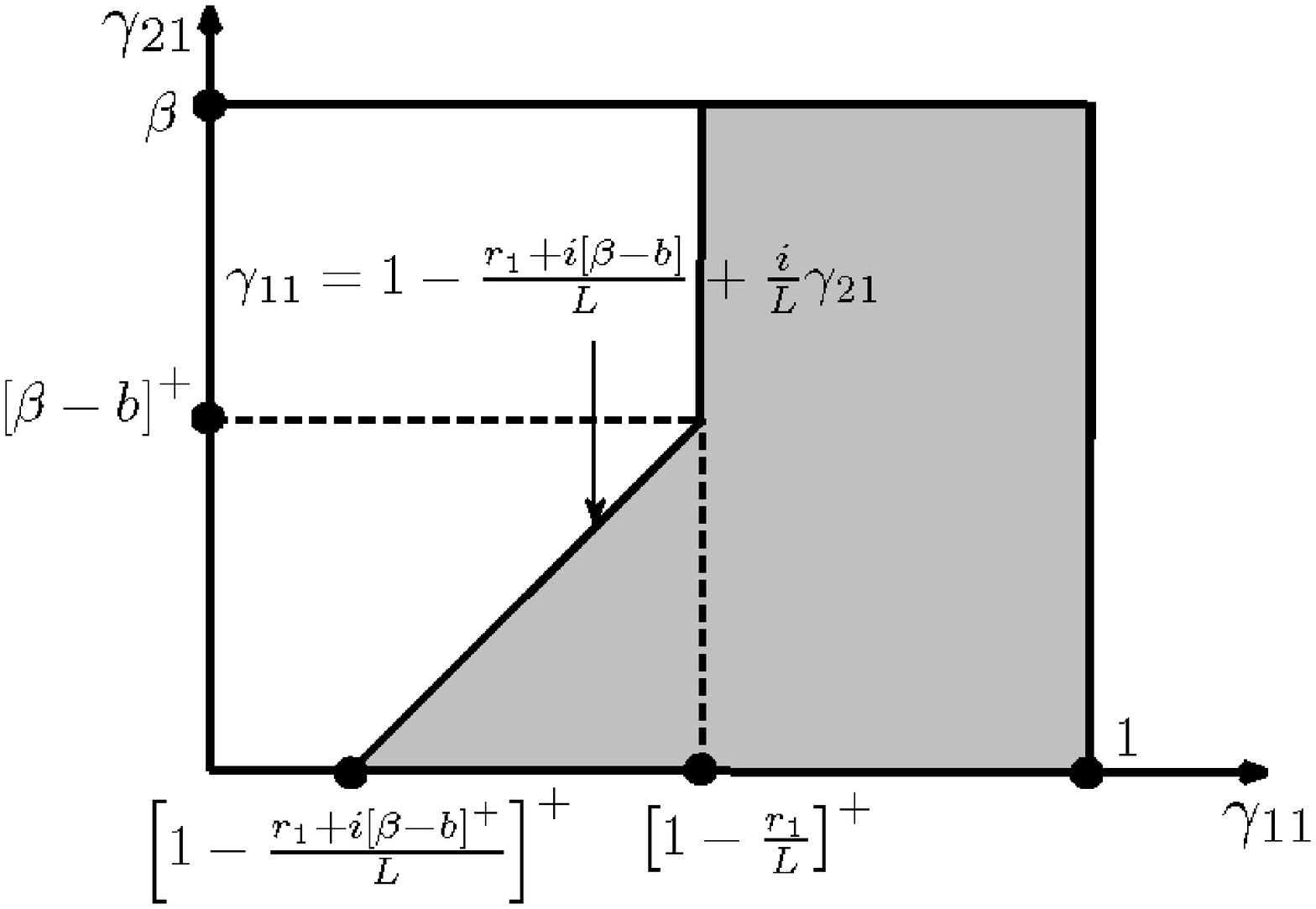}
 }
\subfigure[$r_1<(L-i)(\beta-b)$]{
   \includegraphics[width=75mm, height = 40mm] {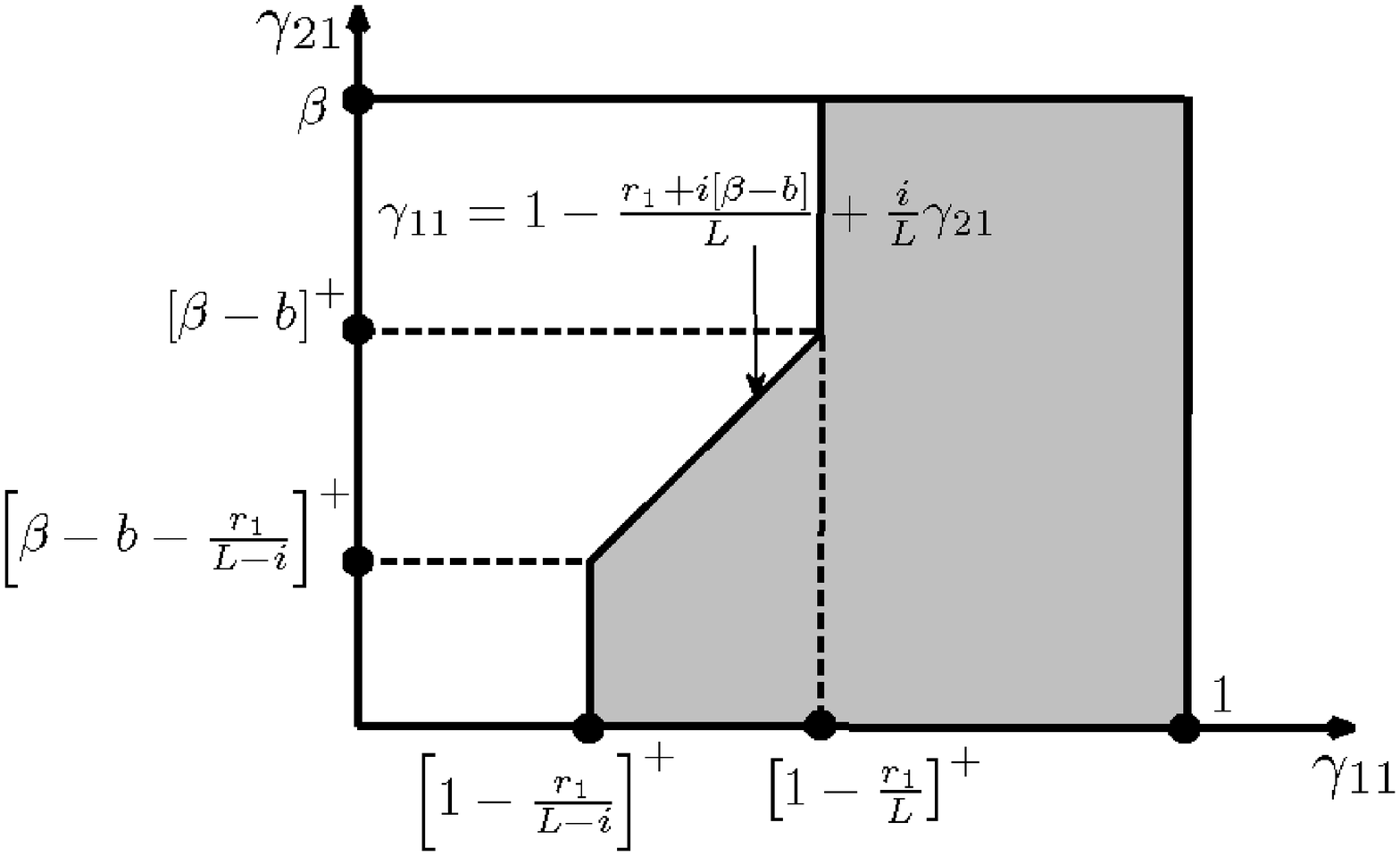}
 }
\caption{Constraint regions of $d_{11,\rm{HK}}(L,i)$}
\label{fig:proof4-a}
\end{figure}

Now, for $d_{12,\rm{HK}}(L,i)$, we have the following minimization problem.
\begin{equation}
d_{12,\rm{HK}}(L,i)=\underset{\gamma_{11},\gamma_{21}\in{\cal{O}}_{12,\rm{HK}}}\min\left\{\gamma_{11}+\gamma_{21}\right\}.
\end{equation}

Similarly, the shaded regions in Fig. \ref{fig:proof4-b} show the constraint regions of $\gamma_{11}$ and $\gamma_{21}$ for $d_{12,\rm{HK}}(L,i)$ considering different values of $r_1$. Thus, we have
\begin{equation}
\begin{split}
&d_{12,\rm{HK}}(L,i)=\\
&\begin{cases}
\left[1-\frac{(r_1+t_2)+i\left[\beta-b\right]^+}{L}\right]^+,\;\;\text{if}\;\; r_1+t_2\geq(L-i)\beta+ib\\
\left[1-\frac{(r_1+t_2)-ib}{L-i}\right]^++\left[\beta-\frac{(r_1+t_2)-ib}{L-i}\right]^+,\\
\qquad\qquad\qquad\qquad\text{if}\;\; Lb<r_1+t_2<(L-i)\beta+ib\\
\left[1-\frac{r_1+t_2}{L}\right]^++\left[\beta-\frac{r_1+t_2}{L}\right]^+,\;\;\text{if}\;\;r_1+t_2\leq Lb.
\end{cases}
\end{split}
\end{equation}

\begin{figure}
\centering
\subfigure[$r_1+t_2\geq(L-i)\beta+ib$]{
   \includegraphics[width=75mm, height = 40mm] {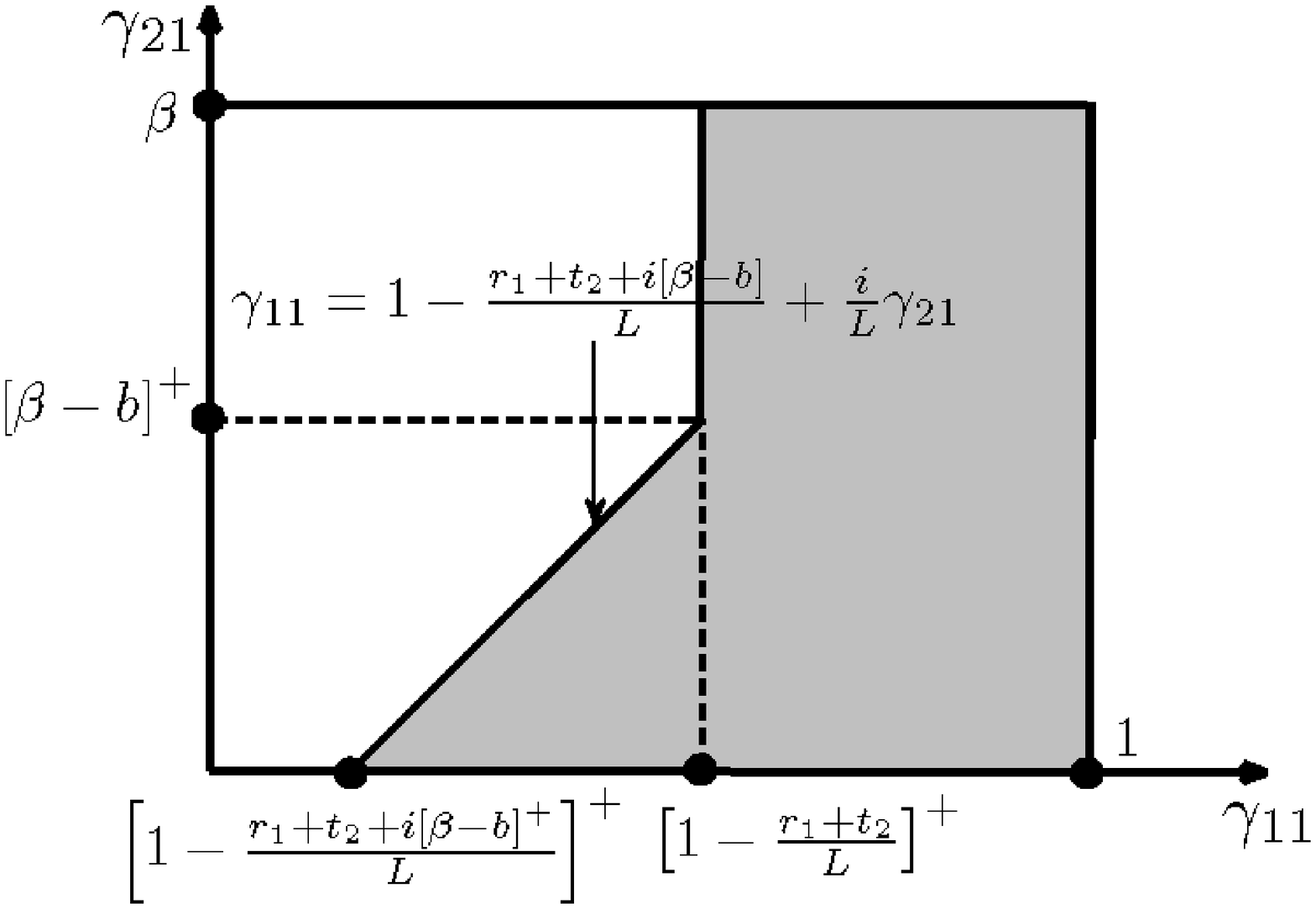}
 }
\subfigure[$Lb\geq r_1+t_2<(L-i)\beta+ib$]{
   \includegraphics[width=75mm, height = 40mm] {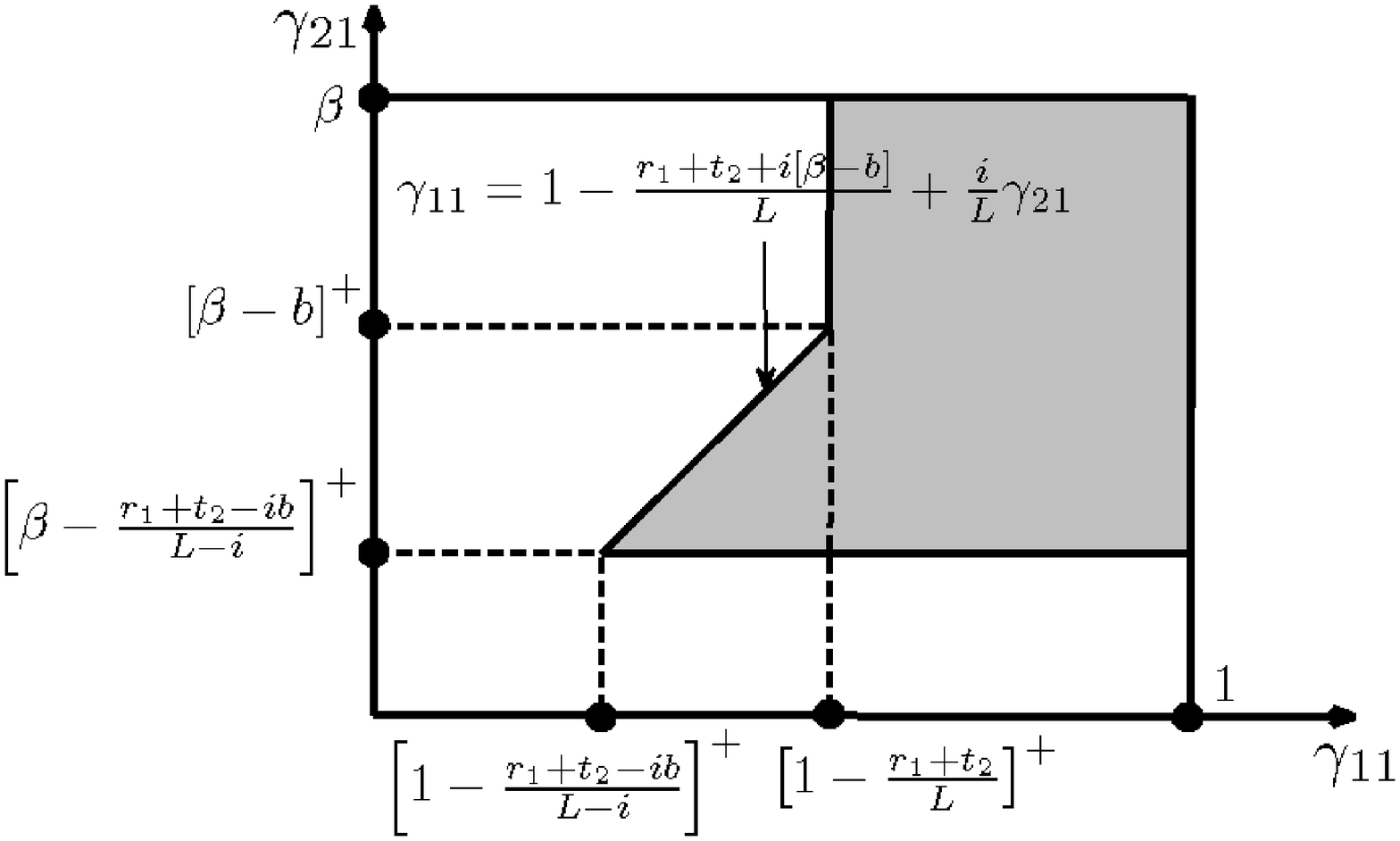}
 }
\subfigure[$r_1+t_2<Lb$]{
   \includegraphics[width=75mm, height = 40mm] {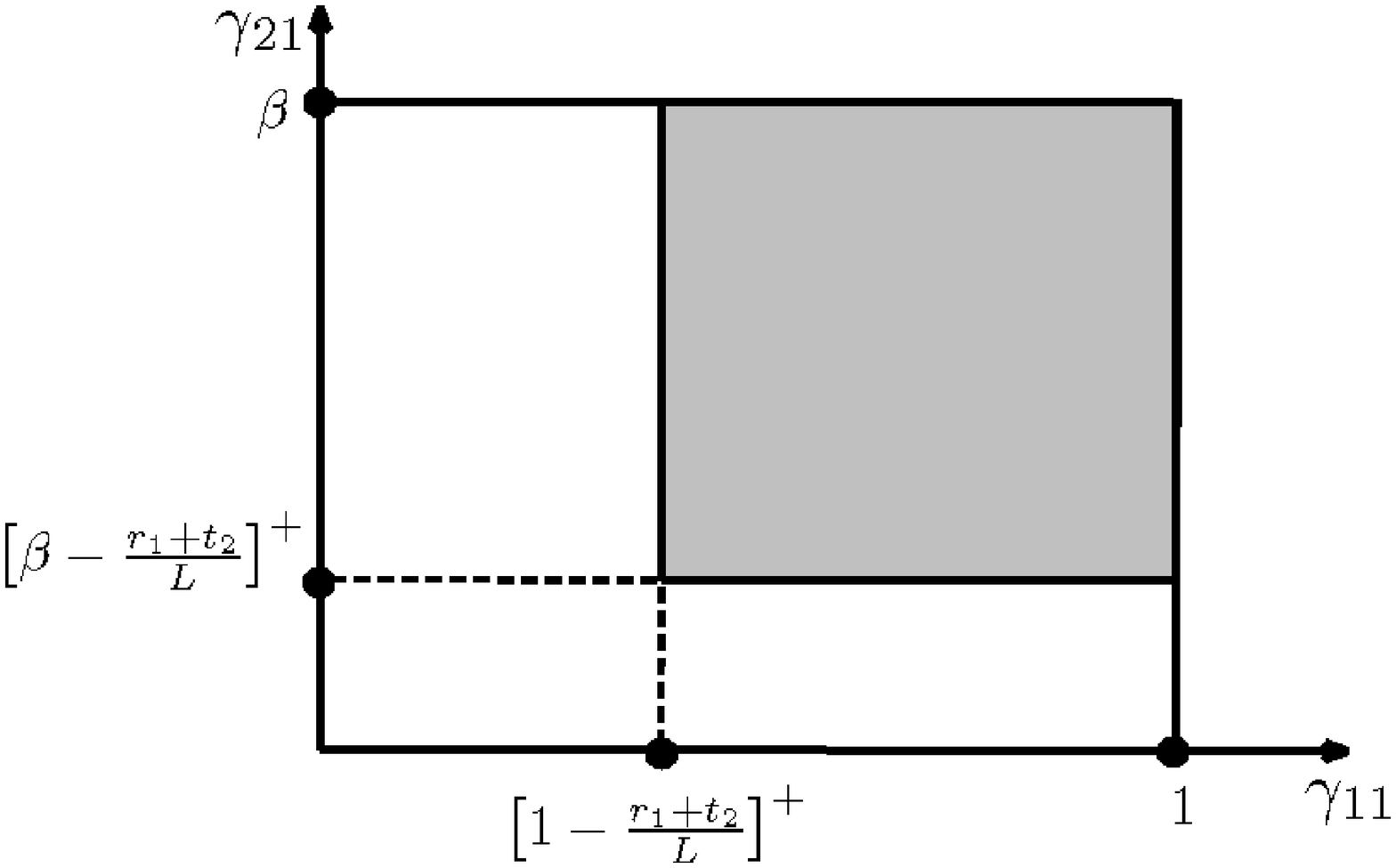}
 }
\caption{Constraint regions of $d_{12,\rm{HK}}(L,i)$}
\label{fig:proof4-b}
\end{figure}

\bibliographystyle{IEEEbib}
\bibliography{MyLib}

\end{document}